%% file: main.tex
\documentclass[11 pt]{article}

\usepackage{helvet}
\usepackage[margin=1in]{geometry}
\usepackage{booktabs} 
\usepackage[ruled]{algorithm2e} 

\SetAlFnt{\small}
\SetAlCapFnt{\small}
\SetAlCapNameFnt{\small}
\SetAlCapHSkip{0pt}
\IncMargin{-\parindent}

\usepackage{amsmath} 
\usepackage{amsthm} 
\usepackage{amssymb}
\usepackage{mathrsfs}
\usepackage{bbm}
\usepackage{bm}
\usepackage{bookmark}
\usepackage{hyperref}
\usepackage{graphicx}
\usepackage{subcaption}
\usepackage[square, comma, sort&compress]{natbib}
\usepackage{setspace}
\usepackage{xcolor}

\newtheorem{theorem}{Theorem}

\newtheorem{lemma}{Lemma}
\newtheorem{proposition}{Proposition}

\newtheorem{remark}{Remark}
\newtheorem{assumption}{Assumption}

\DeclareMathOperator{\EX}{\mathbb{E}}
\DeclareMathOperator{\PR}{\mathbb{P}}
\DeclareMathOperator{\Ind}{\mathbbm{1}}

\title{While Stability Lasts: A Stochastic Model of Non-Custodial Stablecoins
\footnote{This paper is based on work supported by NSF CAREER award \#1653354 and the Bloomberg Fellowship. We thank  Dominik Harz, Georgios Konstantopoulos, the anonymous referees for valuable feedback that helped improve the paper.
}
}

\author{Ariah Klages-Mundt\thanks{Cornell University, Center for Applied Mathematics, Ithaca, NY, 14853, USA, email: {\tt aak228@cornell.edu}.} \ \ \ \ \ \ \ \
Andreea Minca\thanks{Cornell University, School of Operations Research and Information Engineering, Ithaca, NY, 14850, USA, email: {\tt acm299@cornell.edu}.}  \ \ \ \ \ \ \ \
}

\date{June 8, 2022}
\begin{document}

\maketitle

\input{content/1-abstract.tex}

\input{content/2-introduction.tex}

\input{content/3-model.tex}

\input{content/4-foundations.tex}

\input{content/5-domains.tex}

\input{content/6-extensions.tex}

\input{content/7-discussion.tex}

\section{Data availability statement}
The contribution of this paper is theoretical. Where examples have been provided to support theoretical findings, price data is publicly available (by Kaiko - Digital Assets Data Provider and LoanScan platform).

\bibliographystyle{apalike}
\bibliography{references}

\newpage
\appendix

\input{content/9a-generalizing_model.tex}
\input{content/9b-proofs.tex}

\end{document}

%% file: content/1-abstract.tex
\begin{abstract}

The `Black Thursday' crisis in cryptocurrency markets demonstrated deleveraging risks in over-collateralized non-custodial stablecoins.
We develop a stochastic model that helps explain deleveraging crises in these over-collateralized systems. In our model, the stablecoin supply is decided by speculators who optimize the profitability of a leveraged position while incorporating the forward-looking cost of collateral liquidations, which involves the endogenous price of the stablecoin. We formally characterize regimes that are interpreted as stable and unstable for the stablecoin. We prove bounds on quadratic variation and the probability of large deviations in the stable domain and we demonstrate distinctly greater price variance in the unstable domain.
We identify a deflationary deleveraging spiral by means of a submartingale.  
These deleveraging spirals, which resemble short squeezes, lead to faster collateral drawdown (and potential shortfalls) and are accompanied by higher price variance, as experienced on Black Thursday.
 We conclude by discussing non-custodial ways in which the issues raised in this paper can be mitigated.

\end{abstract}

%% file: content/2-introduction.tex
\section{Introduction}

On March 12, 2020, called `Black Thursday' during the COVID-19 market panic, cryptocurrency prices dropped $\sim 50\%$ in the day.\footnote{This occurred while writing up the first draft of this paper.} This was accompanied by cascading liquidations on cryptocurrency leverage platforms, including both centralized platforms like exchanges and new decentralized finance (DeFi) platforms that facilitate on-chain over-collateralized lending. Among many events from this day, the story of Maker's stablecoin Dai stands out, which entered a deflationary deleveraging spiral (akin to a short squeeze on Dai). This triggered high volatility of the `stable' asset and a breakdown of the collateral liquidation process. Due to market illiquidity exacerbated by network congestion, some collateral liquidations were performed at near-zero prices. As a result, the system developed a collateral shortfall, which prompted an emergency response and had to be made up by selling new equity-like tokens to recapitalize \cite{maker_spiral}.

During this time, there was a huge demand for Dai. It became a much riskier and more volatile asset, yet traded at a high premium and fetched lending rates in the mid double digits. Leveraged speculators, who must repurchase Dai in order to deleverage their positions, were exhausting Dai liquidity, driving up the price of Dai and subsequently increasing the cost of future deleveraging (we discuss some further causes that led to market illiquidity in developing the model in the next section). These speculators began to realize that, in these conditions, they face concrete risk that a debt reduction of \$1 could cost a significant premium. Eventually, a new exogenously stable asset--the USD-backed custodial stablecoin USDC--had to be brought in as a new collateral type to stabilize the system \cite{maker_usdc}.

\subsection{Stablecoins}
A stablecoin is a cryptocurrency with added economic structure that aims to stabilize price/purchasing power. For a recent overview of stablecoins, see \cite{klagesmundt2020stablecoins,bullmann19} and the references therein. Stablecoins are meant to bootstrap price stability into cryptocurrencies as a stop-gap measure for adoption. They also serve as mechanics to avoid fiat to
crypto conversions, which are rather costly. This is in fact a key motivation
for their use, hence the system can remain ‘fully decentralized’. 

Stablecoins are either \emph{custodial} and rely on custodians to hold reserve assets off-chain (e.g., \$1 per coin) or \emph{non-custodial} and set up a risk transfer market through  smart contracts, which are programs that execute on the blockchain computer. Custodial stablecoins include Tether, USDC, and the proposed Diem/Libra and can often be viewed analogously to narrow banks or money market funds in terms of underlying structure. Alternatively, non-custodial stablecoins aim to retain the property of reduced counterparty/censorship risk. Figure \ref{fig:daivscustodial} illustrates the market share of the main stablecoins. The largest three are custodial stablecoins (USDT, USDC, BUSD) whereas only one non-custodial stablecoin, Dai, is among the top four stablecoins in terms of market share.
\begin{figure}
	\centering
	\includegraphics[width=0.5\textwidth]{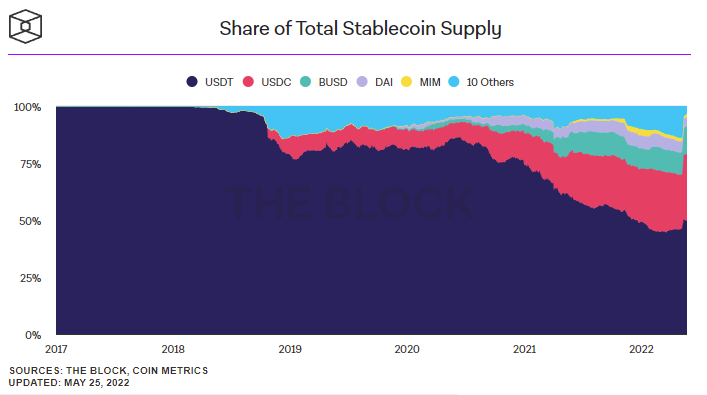}
	\caption{Stablecoin supply.}\label{fig:daivscustodial}
\end{figure}

Non-custodial stablecoins have a wide design space, which is captured in the taxonomy of \cite{klagesmundt2020stablecoins}. A key dimension in this design space is the source of value backing the stablecoin. This ranges from exogenous asset backing, where assets have value unrelated to the system, to endogenous asset backing, where assets are like `system equity' and have value that is circular with the system itself. This latter class, which is often ill-defined as `algorithmic', often blurs the line with being effectively unbacked, as the value of endogenous assets can spiral to zero if confidence is broken.
This latter type includes the Terra UST stablecoin that collapsed in May 2022 \cite{bloomberg2022Terra}.
These stablecoins that are fully or partly endogenously backed can largely be understood using generalizations of currency peg models, such as \cite{morris98}, for which the risks of currency runs and speculative attacks are well studied. These existing tools help to understand these systems and how they (usually) fail, considering that the `economies' around these stablecoins are quite fragile.

In contrast, non-custodial stablecoins that are backed by exogenous assets have greater similarities to non-custodial forms of the current monetary system of commercial bank money, as discussed in \cite{klagesmundt2020stablecoins}. In this paper, we focus on new risks that arise in these types of stablecoins, which require further study.
Stablecoins of this type
transfer risk from stablecoin holders to speculators, who hold leveraged collateralized positions in cryptocurrencies.\footnote{`Leverage' means that speculators holds $>1\times$ initial assets but face new liabilities.} 
The speculator  represents any actor (usually automated) who has an incentive to issue the coin.\footnote{They are part  a form of `keepers' in the MakerDAO protocol.} Such actor issues the stablecoin continuously by locking in collateral. The incentive to issue (or redeem) coin is captured by the speculators'   return expectations including potential liquidation costs and the endogenous stablecoin price.

The collateralization structure is different for non-custodial stablecoins than for the custodial ones. 
 It is similar to a tranche structure, in which stablecoins act like senior debt while speculators are akin the buyers of the junior tranche of a CDO.   In contrast to the classical case,  the `CDO' issue is dynamic and by anyone in the system.  We refer the reader to the Dai  white paper \cite{dai_white}. The white paper describes how {\it anyone could generate Dai using that system}
by leveraging Ethereum (ETH) as collateral through  smart contracts known as Collateralized Debt Positions (CDPs).

A dynamic and automatic deleveraging process balances positions if collateral value deviates too much, as determined by a price feed.  Two major risks in non-custodial stablecoins emerge around market structure collapse and price feed and governance manipulation. In this paper, we focus completely on the market structure risk, assuming that price feeds, governance, and the underlying blockchain perform as expected.\footnote{Note, however, that blockchain congestion can serve to decrease elasticity in the market structure, which we discuss in the model construction.}

In addition to the COVID-19 panic, the effects of these risks are also witnessed in bitUSD, Steem Dollars, and NuBits, which suffered serious depegging events in 2018 \cite{km18}, and Terra and Synthetix, which suffered price feed manipulation attacks in 2019 (\cite{snx1}, \cite{snx2}, \cite{terra}).  Similar manipulations were also observed on the bZx lending protocol in 2020 (\cite{bzx1}, \cite{bzx2}). Many similar examples of mechanism failures and exploitations occurred through the rest of 2020 (see \cite{klagesmundt2020stablecoins, werner2021sok}). Stablecoins currently serve a central role in an increasingly complex decentralized finance environment, involving composability with other DeFi platforms. In addition, many other blockchain assets, such as synthetic and cross-chain assets, rely on the basic mechanism behind stablecoins, which we explore further in the discussion section.

\subsection{This paper}
In this paper, we construct a stochastic model of over-collateralized non-custodial stablecoins, with an endogenous price (Section~\ref{sec:model}). The system is based on a speculator who solves an optimization problem accounting for potential returns from leverage as well as potential liquidation costs. The speculator decides the supply of stablecoins secured by its collateral position while considering demand for the stablecoin.
Our interest in non-custodial stablecoins lies in understanding deleveraging spirals when the price and stablecoin issue is endogenous and the collateral management is decentralized.  In this case, a deleveraging spiral results from the intertwining of a short squeeze in the stablecoin price and a liquidation spiral of the collateral.
This is in contrast to potential liquidation crises in custodial coins such as Tether and USDC or `algorithmic' stablecoins such as Terra UST (which coincidentally also had a partial custodial reserve). Custodial stablecoins maintain stability through arbitrageurs who mint and redeem for assets with the custodian. Unbacked or partially backed stablecoins like Terra UST instead are subject to death spiral risks from runs and speculative attacks due to insolvency.
In both of these cases, classical models for money market funds and currency pegs apply well. \footnote{The recent collapse of the peg in TerraUSD, see e.g., \url{https://www.wsj.com/articles/cryptocurrency-terrausd-falls-below-fixed-value-triggering-selloff-11652122461} can be modeled similarly to the run on money market funds in the financial crisis, \cite{kacperczyk2013safe}, or currency peg attacks, \cite{morris98}. In particular, restoring the peg relies on open market operations by an entity running the reserve fund, such as Luna Foundation in the case of the TerraUSD stablecoins.
}
We focus on the non-custodial variant involving exogenous over-collateralization, whose risks are yet to be analyzed rigorously.

We derive fundamental results about non-custodial stablecoins in our model, including economic limits to the speculator's behavior, in Section~\ref{sec:foundations}. In Section~\ref{sec:domains} we develop the primary results of the paper: we analytically characterize regions in which the stablecoin can be intepreted as stable (Theorems~\ref{result:zprime_doob_ineq} and \ref{result:zprime_qv_pr}) and unstable (Theorems~\ref{result:var_approx} and \ref{result:unstable_var}), and a region in which a deleveraging spiral occurs that can cause liquidity problems in a crisis (Theorem~\ref{result:unstable_submg}). These deleveraging spirals, which resemble short squeezes, are counterintuitive as they lead to stablecoin price appreciation during times of shock, whereas we might otherwise expect prices to depreciate given the riskier state of the system. Further, this appreciation is detrimental: it leads to faster collateral drawdown, and potentially shortfalls, as more collateral is required to fulfill liquidations and is accompanied by higher price variance.

The context for our analytical results is a model with a single speculator facing imperfectly elastic demand for the stablecoin; however, many of the methods can extend to generalized settings. In Section~\ref{sec:perfect_stability}, we consider idealized settings that lead to `perfect' stability properties.

We discuss in Section~\ref{sec:discussion} a seeming contradiction that arises: while the goal is to make decentralized non-custodial stablecoins, these can only be fully stabilized from deleveraging effects by adding uncorrelated assets, which are currently centralized/custodial. This is a consequence of our instability results in Section~\ref{sec:domains} and, as introduced in Section~\ref{sec:perfect_stability}, the absence of a stable region in idealized settings when underlying asset markets deviate from a submartingale setting. We suggest an alternative: a buffer to dampen deleveraging effects without directly incorporating custodial assets. This buffer works by separating those who are willing to have stablecoins swapped to custodial assets in a crisis (in return for an ongoing yield from option buyers) from those who require full decentralization.

Non-custodial stablecoins such as Dai, Rai, and Liquity have since moved in directions such as this to overcome the issues we illustrate in this paper.

\subsection{Relation to Prior Work}
While there is a rich literature on related financial instruments, there is limited research directly applicable to stablecoins.
 \cite{cao18} are the first to point out the analogy of stablecoins to Collateralized Loan Obligations, and contribute to the securitization literature by proposing designs in the decentralized context. They use option pricing theory and  PDE methods for valuation of their new design features.
 Our work is complimentary: we analyze the stability over time of these new securities.

A simple stablecoin model is developed in \cite{km19} and introduces the concept of deleveraging spirals, which later materialized on Black Thursday. This paper supersedes that model and its results. Whereas the model in \cite{km19} doesn't directly account for the actual repurchase price in deleveraging--instead delegating to a risk constraint in the optimization--we set up a stochastic process model in this paper that includes forward-looking liquidation prices in the speculator's optimization. Our analytical results supersede \cite{km19} in the following ways:
\begin{itemize}
\item We formally characterize a deleveraging spiral as a submartingale, whereas their paper lacks a formal treatment.
\item We give stability results in terms of probabilities of large deviations and quadratic variation of the price process.
\item An unstable region is conjectured in their paper, backed by simulation. We formally prove distinct price variances in stable and unstable regions.
\end{itemize}

\cite{evans2019ratings} analyzes credit risk stemming from collateral type in Maker's stablecoin Dai. \cite{terra2, celo} model stability in Terra and Celo stablecoins under Brownian motion scenarios in the absence of endogenous market feedback effects that motivate this paper. 
\cite{klagesmundt2020stablecoins, huo2022decentralized} discuss models of governance and oracle attack surfaces  for non-custodial stablecoins.  More generally in the context of decentralized finance, \cite{werner2021sok} treat the governance extractable value.

\cite{detrio15} discusses stablecoin concepts based on monetary policy and hedging strategies and introduces methods for enhancing liquidity using combinatorial auctions and automated market makers. \cite{lipton18} studied custodial stablecoins and considers the use of hedging techniques to build an asset-backed cryptocurrency. \cite{gudgeon20} explores the robustness of decentralized lending protocols to shocks and liquidations. \cite{chitra20} explores competition between decentralized lending yields and staking yields in proof-of-stake blockchains. However, these do not model a stablecoin mechanism with endogenous price behavior.

\cite{harz19} designs a reputation system for crypto-economic protocols to reduce collateral requirements. This does not readily apply to understanding stablecoin collaterals, however, as it requires identification of `good' behavior and, additionally, stablecoin speculators face leveraged exchange rate bets and will have reason to provide greater than minimal collateral. This additionally motivates our model to understand how liquidation effects affect speculator decisions.

Stablecoins share similarities with currency peg models, e.g., \cite{morris98, guimaraes03}. In these models, the government plays a mechanical market making role to seek stability and is not a player in the game.  In contrast, in non-custodial stablecoins, decentralized speculators take the market making role. They issue/withdraw stablecoins to optimize profits and are not committed to maintaining a peg. In a stablecoin, the best we can hope is that the protocol is well-designed and that the peg is maintained with high probability through incentives. A fully strategic model would be a complicated (and likely intractable) dynamic game.

There are also similarities with collateral and debt security markets and repurchase agreements. These have also experienced unprecedented stress in the COVID-19 market panic, during which even 30-year US government bonds--normally highly liquid--have been difficult to trade \cite{ft_bonds}. Such debt securities differ from stablecoins in that dollars are borrowed against the collateral as opposed to a new instrument, like a stablecoin, with an endogenous price. These debt security markets do, however, demonstrate that liquidity in the underlying markets can dry up in crises even in highly liquid markets. Stablecoins face this liquidity risk in the underlying market as well as an endogenous price effect on the stable asset.

The problem resembles classical market microstructure models (e.g., \cite{ohara95}); it is a multi-period system with agents subject to leverage constraints that take recurring actions according to their objectives. In contrast, the stablecoin setting has no exogenously stable asset that is efficiently and instantly available. Instead, agents make decisions that endogenously affect the price of the `stable' asset and affect future incentives.

%% file: content/3-model.tex
\section{Model}\label{sec:model}

Our model is very closely related to Maker's stablecoin Dai \cite{dai_white} as well as newer stablecoins by UMA, Reflexer, and Liquity.  Crucially, these stablecoins are backed by over-collateralization in assets that have value exogenous to the stablecoin system as opposed to assets whose value is circularly derived from the stablecoin itself.
There are two primary feedback effects to consider in these stablecoins: (1) feedback of deleveraging on an endogenous stablecoin price, and (2) feedback of deleveraging on collateral price. We focus on the former. The latter can be described using existing deleveraging models (e.g., this is considered in the stablecoin context in \cite{gudgeon20}). We later discuss how our model can be adapted to incorporate these endogenous effects on collateral in Section~\ref{sec:discussion}.

The model contains a stablecoin market and two assets: a risky asset (ETH)\footnote{We designate the risky collateral asset as ETH for simplicity. In principle, it could be another cryptoasset or even outside of a cryptocurrency setting.} with exogenous price $X_t$ and an ETH-collateralized stablecoin STBL with endogenous price $Z_t$. The stablecoin market connects stablecoin holders, who seek stability, and speculators, who make leveraged bets backing STBL. The STBL protocol requires the STBL supply to be over-collateralized in ETH by collateral factor~$\beta$.

In order to focus on the effects of speculator decisions in this paper, we simplify the stablecoin holder demand as exogenous with constant unit price-elasticity. This is equivalent to a fixed STBL demand $\mathcal D$ in dollar terms, though not quantity. 
Note that there is no direct redemption process for stablecoin holders aside from a global settlement/shutdown of the system at par value, which can be triggered by a governance process (see \cite{dai_white}).

From a practical perspective, STBL demand is not elastic, at least short-term, even if it were in principle elastic longer-term. A significant portion of stablecoin supplies are locked in other applications, like lending protocols and lotteries. These applications promise (in some sense) value safety in over-collateralization, but don't guarantee liquidity to withdraw. Additionally, Ethereum transactions cannot be executed in parallel; during volatile times, transactions can be delayed due to congestion, causing timely trades (especially involving transfer to/from centralized exchanges) to fail. This occurs even if, in principle, there is liquidity in these markets. On the other hand, longer-term demand elasticity will naturally depend on the presence of good uncorrelated alternatives.\footnote{From another perspective, a strategic stablecoin holder would take into account expectations about speculator issuance and ability to maintain the price target and expectations about a global settlement. This is outside of our model as formulated.}

The speculator has ETH locked in the system and decides the STBL supply, which represents a liability against its locked collateral. At the start of step $t$, there are $\mathcal L_{t-1}$ STBL coins in supply. The speculator holds $N_{t-1}$ ETH and chooses to change the STBL supply by $\Delta_t = \mathcal L_t - \mathcal L_{t-1}$. If $\Delta_t > 0$, the speculator sells new STBL on the market for ETH at the market clearing price $Z_t$. This increases the ETH position $N_t$. If $\Delta_t < 0$, the speculator buys STBL on the market, reducing $N_t$.  We denote by $\bar N_t$ the speculator's locked collateral. Informed by limitations of actual implementations, we formalize the process $(\bar N_t)$ based on $(N_t)$.\footnote{In principle, the speculator's decision could be extended to deciding $\bar N_t$ in addition to $\Delta_t$. Note however that this would make most sense if the speculator's position is further extended to include multiple assets.} The speculator decides $\mathcal L_t$ by optimizing expected profitability in the next period based on expectations about ETH returns and the cost of collateral liquidation if the collateral factor is breached.

In this way, the speculator myopically optimize for the next period. A simplification of our model is a one-off game, which hosts a single period of decision-making before the system is settled in the final period.
In this case, the myopic setup is parallel to major single period games in finance (e.g., \cite{morris98, guimaraes03, diamond1983bank, dybvig1991capital, parlatore2016fragility}).
Our results make significant contributions over the existing state of research on stablecoins, describing different system behavior depending on initial conditions in one-off games.
The more general multi-period form of our model then describes a dynamic process composed of a series of one-off games with changing initial conditions.
Our results also apply more generally to this multi-period setting, where they are stronger than simply a series of the one-off version of the results.
Both of these contribute to stablecoin modeling as there are not better candidates for multi-period models at this point, although we later discuss ideas toward adapting the model into a multi-period control problem.

Given supply and demand, the STBL market clears by setting demand equal to supply in dollar terms. This yields the clearing price $Z_t = \frac{\mathcal D}{\mathcal L_t}$.\footnote{We can consider constant elasticity STBL demand functions that depend on $Z_t$.  Letting $q$ be the quantity  of STBL demanded at \$1 price and assuming a constant price elasticity $-\gamma<0$, the dollar-denominated demand function is
$\mathcal D(Z_t) = Z_t Q(Z_t) = Z_t q/(1-\gamma (1-Z_t)).$ for $\gamma=1$ we obtain the case of constant dollar denominated demand. In clearing the market, the generalized price process is a linear transformation
$Z_t = \frac{1}{\gamma}\left( \frac{q}{\mathcal L_t} -1\right) + 1$.
} This clearing equation is related to the quantity theory of money and is similar to the clearing in automated market makers \cite{angeris20} but processed in batch.

\subsection{Formal setup}
We formalize the model as follows. We define the following \emph{parameters}:
\begin{itemize}
\item $\mathcal D$ = STBL demand in dollar value (equivalent to constant unit price-elasticity)
\item $\beta$ = collateral factor for ETH
\item $\alpha \geq 1$ = liquidation cost multiple (reflecting the fee paid to liquidators)
\end{itemize}
The system is composed of the following \emph{processes}:
\begin{itemize}
\item $(X_t)_{t\geq 0}$ = exogenous ETH price process in dollars.
\item $\mathcal L_t$ = stablecoin supply at time $t$ that obeys
	$$\mathcal L_t = \zeta + L_{t-1} + \Delta_t,$$
	where $L_{t-1}> 0$ is the speculator's STBL liabilities from the previous period, $\Delta_t$ is the speculator's change in liabilities at time $t$ (such that $L_t = L_{t-1} +\Delta_t$), and $\zeta$ is a real number that modifies circulating supply
\item $N_t$ = speculator's ETH position at time $t$, including collateral
\item $\bar N_t$ = speculator's locked ETH collateral at time $t$ (and start of time $t+1$)
\item $(Y_t)_{t\geq 0}$ = speculator's value process
\item $Z_t = \frac{\mathcal D}{\mathcal L_t}$ defines the STBL price process.
\end{itemize}
We take $(\mathcal{F}_t)_{t\geq 0}$ to be the natural filtration where $\mathcal{F}_t = \sigma(X_0,\ldots,X_t, \mathcal{L}_0,\ldots, \mathcal{L}_t)$. The system is driven by the process $(X_t)$ subject to the speculator's decisions $\Delta_t$ (equivalently $L_t$ given $L_{t-1}$).

The parameter $\zeta$ modifies circulating STBL supply. This could come from an outside amount of STBL not created by the speculator (a positive adjustment), or some STBL could essentially be locked (a negative adjustment). As formulated, our model applies to a system that can be described with monopolistic agents, or where agents behave similarly (have similar beliefs). With $\zeta > 0$, the model becomes similar to having heterogeneous agents. Whereas, in general to do this, we would have to consider both heterogeneous beliefs about the future as well as different $\zeta$s, which together would be intractable, $\zeta$ provides a way to aggregate these various effects in a simpler model. In particular, we suggest a positive $\zeta$ may make numerical results more applicable to real settings.

To simplify the exposition of analytical results going forward, we simplify to the case that $\beta=\frac{3}{2}$ (the collateral factor used in Maker's Dai stablecoin) and $\zeta=0$. \emph{Note that under these conditions, and in the remainder of the paper, we use $L_t$ and $\mathcal L_t$ interchangeably}.


\subsection{Collateral constraint}
The collateral constraint requires the collateral locked in the system to be $\geq$ a factor of $\beta$ times by liabilities. It applies in both a pre-decision and post-decision sense. The \emph{pre-decision} version determines when a liquidation occurs: a liquidation is triggered at the start of time $t$ if the following condition is breached
$$\bar N_{t-1} X_t \geq \beta L_{t-1}.$$
The \emph{post-decision} version constrains the speculator's decision-making, limiting $L_t$ such that
$$\bar N_t X_t \geq \beta L_t.$$
Note that the nominal stablecoin price (\$1) is used in these constraints instead of the real price because these are encoded by the protocol's smart contracts as one of the means toward incentivizing the \$1 target.\footnote{Conceptually, outside of this model, this has the effect of upper bounding the stablecoin price at $\beta$ as an arbitrage opportunity would be created otherwise.}
The collateral factor could be dynamic, in the sense that the governance of the protocol could vote to change its value. 
Proposals to change the collateral factor are in practice infrequent, see \url{https://makerdao.world/en/learn/vaults/liquidation/}, so we consider here a constant factor. We leave it for future research to model the governance's decision.
\subsection{Speculator decides $\Delta_t$ taking into account real liability value}
We assume the speculator is risk-neutral and optimizes its next-period expected value, taking into account expectations around liquidations. In particular, this means that the speculator takes into account the real cost of deleveraging its liabilities in the event it needs to reduce its position in the next time step and doesn't simply measure the nominal value of liabilities.
Its value at time $t$ is its nominal equity at the start of period (pre-decision), adjusted by a liquidation effect that describes how the real value deviates from nominal in the event that the speculator needs to deleverage. That is
$$Y_t = N_{t-1} X_t - L_{t-1} - \text{liquidation effect}.$$
A liquidation effect is outlined in a following subsection.

Note that $N_t$ is a function of the decision variable $\Delta_t$, and recall $L_t = L_{t-1} + \Delta_t$. The speculator decides $\Delta_t$ (equivalently $L_t$ given $L_{t-1}$) to optimize next-period expected value subject to the post-decision collateral constraint in the current period:
$$\begin{aligned}
\max_{\Delta_t} \hspace{1cm} & \EX [Y_{t+1}|\mathcal{F}_t] \\
\text{s.t.} \hspace{1cm} & \bar N_t X_t \geq \beta L_t.
\end{aligned}$$

Thus the speculator accounts for the expected deviation of real from nominal liability value. If the expected liquidation effect is small ---for instance if the probability that the speculator needs to deleverage next period is small--- then the speculator treats $L_t$ near face value in the optimization for a mix of short- and long-term reasons.
As long as speculators can survive  liquidation, they can expect to dispose of liabilities near face value longer-term when markets are liquid. The protocol smart contracts also add a precedent for treating liabilities at face value: it is encoded in this way in the collateral constraint and in the event of global settlement of the system, which is intended to be be triggered should the system diverge too significantly from the intended structure (and which would occur in the final period of the one-off version).

\subsection{Speculator's collateral at stake}
We consider that the speculator decides on a level of participation as a component of their entire portfolio. This takes place in a separate optimization problem outside the scope of this model (although we discuss how it could be extended later). The speculator's level of participation amounts to the initial collateral at the start of our model--for simplicity, we say this also includes any amount they have decided beforehand may be accessible to top up collateral later.
The speculator's behavior in our model amounts to maximizing the expected value of this component of their portfolio.
On the other hand, if this were the speculator's entire portfolio, we note that the story may be different--e.g., they may want to maximize expected log values as in the Kelly criterion and would probably choose to participate differently, as is common in problems of leverage if the whole portfolio is at stake.

We take the speculator's collateral at stake at the start of time $t+1$ to be $\bar N_t = N_{t-1}$ minus any collateral liquidation that happens at time $t$. This is consistent with the speculator's collateral being blocked: it cannot be used to repurchase STBL in the same step. This means that the speculator (1) has an outside amount (or is able to borrow) to repurchase STBL if $\Delta_t<0$ and then later repays this from unlocking collateral and (2) can't post proceeds of new STBL issuance ($\Delta_t>0$) as collateral within the same step.

While there are settings in which we could alternatively use $N_t$ as the collateral at stake at the start of $t+1$ (e.g., if flash loans are used), the choice of $N_{t-1}$ additionally leads to a simpler exposition of results as it decouples the collateral from the decision variable. 

\subsection{Collateral liquidation mechanics}
In time $t+1$, the pre-decision collateral constraint is $\bar N_t X_{t+1} \geq \beta L_t$. If this is breached, then the speculator's collateral is partially liquidated, if possible, to repurchase an amount $\ell_{t+1}>0$ of STBL.
In real protocols, liquidation amounts are automated by an algorithm and will inherently be first order estimates of the amount needed to rebalance the debt position as the algorithm will not be able to know the actual market structure and price impact.
For instance, liquidations in Maker and Compound release a certain amount of debt to be repaid,  and unlock a corresponding amount of collateral that an arbitrager can use to rebalance the debt position (both decided algorithmically in  \cite{comp_white} and  \cite{dai_white}, and the latter decided through auction in Maker's newer version \cite{dai_white_new}).
Consistent with these protocols, we set the amount of debt that needs to be repaid in a liquidation to be $\ell_{t+1}$ of STBL such that post liquidation we have $\bar N_t X_{t+1} - \ell_{t+1} = \beta(L_t-\ell_{t+1})$. With $\beta=\frac{3}{2}$, this amount is
$$\ell_{t+1} = \frac{\beta L_t - \bar N_t X_{t+1}}{\beta-1} = 3L_t - 2 \bar N_t X_{t+1}.$$
We interpret this as the protocol's encoded estimate, using nominal stablecoin price, of how much collateral it should liquidate in an `auction' to deleverage, similar to Maker. Our model simplifies the auction to settle on the endogenous stablecoin market.
Other liquidation algorithms could also be considered and would lead to similar qualitative effects. 

In a time step with a liquidation, the liquidation forces an upper bound $\Delta_{t+1} \leq -\ell_{t+1}$ as this amount would, in the real protocol, be unlocked for arbitrageurs. But the speculator could choose to repurchase more STBL to further reduce leverage. The repurchase of $\ell_{t+1}$ through the liquidation mechanism is subject to a liquidation cost multiple $\alpha\geq 1$--i.e., the effective repurchase price is $\alpha\times$ the STBL market price. The purpose of this fee is that, in real stablecoin systems,  liquidations are performed by arbitrageurs who capture this fee.

Notice that the STBL market price will itself be affected by liquidations. Depending on market impact, which the algorithms can only observe sequentially, the liquidation may be insufficient to fully rebalance the debt position back to the collateral constraint. If this occurs, then the issue will be taken into account with further liquidations in subsequent time steps. The parameter $\beta$ in real systems is intended to provide safety in such events so that the system does not become under-collateralized.

Two thresholds are relevant at time $t$ for calculating expectations of a liquidation effect at time $t+1$. These are non-time-dependent functions of the random variable $L_t$:
$$b(L_t) := \frac{\beta L_t}{ \bar N_t}$$
$$c(L_t) := \frac{1}{2\bar N_t} \Big(\sqrt{\alpha^2 \mathcal{D}^2 + 4\alpha\mathcal{D}L_t + L_t^2} - \alpha\mathcal{D} + L_t\Big).$$
The threshold $b(L_t)$ gives the highest $t+1$ ETH price that breaches the collateral constraint while the threshold $c(L_t)$ gives the $t+1$ ETH price that consumes the entirety of the speculator's locked collateral in a liquidation repurchase due to the effect on STBL repurchase price.\footnote{The probability of a large deviation like this is not zero. For instance, it could represent the possibility of a contentious hard fork that splits ETH value.} Below this level, the speculator cannot meet the collateral demand even by liquidating everything. The formulation of $b(L_t)$ follows directly from the collateral constraint; the formulation of $c(L_t)$ follows from equating the repurchase cost of liquidation $\ell_{t+1}$ to $\bar N_t X_{t+1}$ and solving for $X_{t+1}$.

If $c(L_t) \leq X_{t+1} \leq b(L_t)$, then the liquidation effect is $\ell_{t+1} - \ell_{t+1} \frac{\mathcal{D}}{\mathcal{L}_t - \ell_{t+1}}\alpha$. This represents a repurchase of $\ell_{t+1}$ STBL (reducing collateral by the repurchase price $\frac{\mathcal D}{\mathcal L_t - \ell_{t+1}}$ with liquidation fee factor $\alpha$) and subsequent reduction of the speculator's liabilities by the $\ell_{t+1}$. The variables $\mathcal L_{t+1}$ and $N_t$ are affected similarly.\footnote{Note that $N_t$ is affected because this is the locked collateral at time $t+1$. Alternatively, working with $N_{t+1}$ as locked collateral, we would update $N_{t+1}$.} If $X_{t+1}< c(L_t)$, then the speculator's collateral position is zeroed out in the liquidation. We define the corresponding events
$$A_t = \{ X_{t+1} \geq b(L_t) \}$$
$$B_t = \{ c(L_t) \leq X_{t+1} < b(L_t) \}.$$

\subsection{System of random variables}

Putting all the pieces together, we have the following system of random variables driven by the random process $(X_t)$:

$$\begin{aligned}
X_t & \\
Y_{t+1} &= \frac{\Delta_t \mathcal D X_{t+1}}{\mathcal L_t X_t} + (\bar N_t X_{t+1} - L_t)\Ind_{A_t \cup B_t} + \Ind_{B_t}(3L_t - 2\bar N_t X_{t+1} )\Big(1- \frac{\alpha\mathcal D }{2 \bar N_t X_{t+1} - 2L_t}\Big) \\
\Delta_t^* &= \begin{cases}
\min\Big( \arg \max_{\Delta_t} \EX [Y_{t+1}| \mathcal{F}_t], \frac{\bar N_{t-1}X_{t}}{\beta}-L_{t-1}\Big) & \text{ if } X_t \geq \frac{\beta L_{t-1}}{\bar N_{t-1}} \\
\min \Big( \arg \max_{\Delta_t} \EX [Y_{t+1}| \mathcal{F}_t], - (3\mathcal{L}_{t-1} - 2 \bar N_{t-1} X_t)\Big) & \text{ if } X_t < \frac{\beta L_{t-1}}{\bar N_{t-1}}
\end{cases} \\
    \mathcal{L}_t &= \mathcal{L}_{t-1} + \Delta_t^* \\
N_t &= \begin{cases}
N_{t-1} + \Delta_t^* \frac{Z_t}{X_t} & \text{ if } X_t \geq \frac{\beta L_{t-1}}{\bar N_{t-1}} \\
N_{t-1} + \frac{Z_t}{X_t}(\Delta_t + (1-\alpha)(3\mathcal{L}_{t-1}-2 \bar N_{t-1}X_t)) & \text{ if } X_t < \frac{\beta L_{t-1}}{\bar N_{t-1}}
\end{cases}\\
\bar N_t &= \begin{cases}
N_{t-1} & \text{ if } X_t \geq \frac{\beta L_{t-1}}{\bar N_{t-1}} \\
N_{t-1} - \alpha (3\mathcal L_{t-1} - 2 \bar N_{t-1} X_t) & \text{ if } X_t < \frac{\beta L_{t-1}}{\bar N_{t-1}}
\end{cases} \\
    Z_t &= \frac{\mathcal{D}}{\mathcal{L}_t}.
\end{aligned}$$

In the above, the first case for $\Delta_t^*$ comes from maximizing expected value subject to the post-decision collateral constraint while the second cases for $\Delta_t^*$, $N_t$, and $\bar N_t$ apply the liquidation effects that occur during time $t$.

%% file: content/4-foundations.tex
\section{Foundational Results}\label{sec:foundations}

In this section, we derive foundational results about the model that we will use to prove the primary results of the paper in the next section.

\subsection{Assumptions}
\label{sec:assumption}
We begin by defining the assumptions we will use in the rest of the paper.

\begin{assumption}
$(X_t)$ is a submartingale with respect to $(\mathcal F_t)$ and is independent from $(\mathcal L_t)$ and $(N_t)$.
\end{assumption}

A submartingale is a stochastic process in which the expected future value, conditioned on all prior values, is greater than or equal to the current value. The submartingale assumption can be relaxed somewhat while preserving some results. It is useful, though not necessarily critical, in our proof of problem concavity. However, the results are most meaningful in a setting like a submartingale, which always provides a fundamental reason that a speculator might desire leverage. In such a setting, it is \emph{conceivable} that the stablecoin could maintain a dollar peg, whereas in long periods of negative expected returns, the stablecoin concept falls apart as no speculators will want to participate. As noted in the introduction, such a deviation from the submartingale setting appears to have occurred in March 2020. 

\begin{assumption}
Each $X_{t+1}$ has a conditional probability distribution given $\mathcal F_t$, which admits a density function $f_t$ that is continuous almost surely.
\end{assumption}

Equivalently, we consider the process in terms of returns $R_t$, where $X_{t+1} = X_t R_{t+1}$. Conditioned on $\mathcal F_t$, then $R_{t+1}$ admits density function $g_t$. In the i.i.d. setting for $(R_t)$, the time dependence can be dropped. For most results, we do not need to assume i.i.d.

\begin{assumption}
There is some upper bound $r \geq \sup_n \EX[R_n | \mathcal F_{n-1}]$.
\end{assumption}

The next assumption is needed to interchange derivative and integration operators. It also translates to an upper bound on $\mathcal L_t$ and a lower bound on $N_{t-1}$.
\begin{assumption}
There is some upper bound $u \geq c(L_t)$ for all $L_t$.
\end{assumption}

The next assumption ensures that the STBL price is bounded away from infinity.
\begin{assumption}
$\mathcal L_t \geq v > 0$ for some $v$.
\end{assumption}

The next assumption simplifies repurchase considerations. It is reasonable given a reasonable bound $r$ on expected returns.
\begin{assumption}
The liquidation premium factor $\alpha$ is sufficiently high that the repurchase price in a liquidation is  $>1$ almost surely.
\end{assumption}

The next assumption translates to a reasonable condition on $X$ distributions considering $b(L_t)$ is  linearly increasing whereas $c(L_t)$ decreases with $L_t$.
\begin{assumption}
$\PR(B_t | \mathcal F_t) = \PR\Big(c(L_t) \leq X_{t+1} \leq b(L_t) | \mathcal F_t\Big)$ is increasing in $L_t$.
\end{assumption}

Define $\psi(\mathcal L_t) := \EX[Y_{t+1} | \mathcal{F}_t]$. Note that $\psi$ could have a subscript $t$, or equivalently other time $t$ inputs ($\bar N_t, X_t, g_t$), but we relax notation as we only use it in the context of time $t$. The next assumption ensures that $\psi$ is concave in $\mathcal L_t$, a result that we prove in Proposition ~\ref{result:exy_concave}. When this is not met, the model starts in a strange region in which the speculator's objective can be non-concave and real and nominal liability values can be disassociated. This is an artifact of the simplified structure of demand in the model, which we would expect to adapt in such a setting. Thus we expect the model to not apply well outside of this assumption. Live stablecoin systems that remain operational readily satisfy this assumption.
\begin{assumption}
$\frac{\alpha \mathcal D N c_t}{2(Nc_t - \mathcal L_t)^2} \leq 2$ (note $\mathcal L_t \geq \frac{27}{46} \alpha \mathcal D$ (or $  \alpha Z_t \leq \frac{46}{27} $  is sufficient).
\label{assumption:real-sim-nominal}
\end{assumption}
 Live stablecoin systems  readily satisfy this assumption.\footnote{Recall that  $\alpha \geq 1$ is the liquidation cost multiple (reflecting the fee paid to liquidators). Assuming  $\alpha = 1.05$, the sufficient condition in  Assumption \ref{assumption:real-sim-nominal} is implied by $Z_t < 1.62$, which is verified in practice for all live stablecoins.}

Additionally, the next assumption ensures that $\psi$ is \emph{strictly} concave in $\mathcal L_t$, which we also prove in Proposition~\ref{result:exy_concave}. Notice that this means that \emph{either} the submartingale inequality is strict at time $t$ or there is non-zero probability that a liquidation is triggered in the next step. Given that the latter is certainly reasonable, this assumption is not much stronger than the basic submartingale assumption.
\begin{assumption}
Either $\EX[R_{t+1} |\mathcal F_t] > 0$ or $\PR(B_t | \mathcal F_t) = \PR\Big(c(L_t) \leq X_{t+1} \leq b(L_t) | \mathcal F_t \Big) > 0$.
\end{assumption}

While strict concavity of $\psi$ is not necessary for all results, it does simplify the analysis considerably. More generally, concavity of $\psi$ could reasonably be expected in many settings, and so the assumptions can probably be relaxed. Informally, reasonable distributions for $X_t$ will have concentration about the center. In this case, moving $\Delta_t$ in the positive direction, expected liabilities increase faster than revenue from new STBL issuance. Moving $\Delta_t$ in the negative direction, the cost to buyback grows faster than the decrease in expected liabilities.

\subsection{Concavity and scale invariance}

Our first result is to prove that $\psi(\mathcal L_t)$ is concave in $\mathcal L_t$.
\begin{proposition}\label{result:exy_concave}
Given Assumptions 1-8, $\psi(\mathcal L_t) := \EX[Y_{t+1} | \mathcal{F}_t]$ is concave in $\mathcal L_t$.

Further, given additional Assumption 9, $\psi(\mathcal L_t)$ is \emph{strictly} concave in $\mathcal L_t$.
\end{proposition}

\begin{center} \hyperlink{pf:exy_concave}{\texttt{[Link to Proof]}} \end{center}

In deriving some results, it will be useful to make assumptions about the scale of the system. The next result shows that results about $Z_t$ should translate to differently scaled systems, validating that such results will describe the STBL price process more generally. In the following, we define $h$ to output $\mathcal L_t$ as a function of the system state.

\begin{proposition}\label{result:rescalings}
Consider a system setup $(L_{t-1},\mathcal D,N_{t-1})$ with ETH price process $(X_t)$.  For $\gamma>0$,
$$h(\gamma L_{t-1}, \gamma \mathcal D, \gamma N_{t-1}, X_t) = \gamma h(L_{t-1}, \mathcal D, N_{t-1}, X_t)$$
$$h(L_{t-1}, \mathcal D, \frac{1}{\gamma} N_{t-1}, \gamma X_t) = h(L_{t-1}, \mathcal D, N_{t-1}, X_t).$$
As a result, the STBL price process $(Z_t)$ is equivalent across these system rescalings.
\end{proposition}

 \begin{center} \hyperlink{pf:rescalings}{\texttt{[Link to Proof]}} \end{center}

Under these condtions, we can interchange derivative and integration operators in $\frac{\partial \psi}{\partial \mathcal L_t}$ according to Leibniz integral rules (a variation of dominated convergence theorems). The speculator's choice of $\mathcal L_t$ will fulfill the first order condition of $\frac{\partial \psi}{\partial \mathcal L_t} =0$. From concavity, we can then conclude that the speculator chooses to increase the STBL supply when $\frac{\partial \psi}{\partial \mathcal L_t} (\mathcal L_{t-1})>0$ and to decrease the STBL supply when $\frac{\partial \psi}{\partial \mathcal L_t} (\mathcal L_{t-1})<0$.

Note that we can derive sufficient conditions for these events using Lemma~\ref{prop:elem_ineq} from the Appendix. Such conditions can be useful as concrete interpretations of the events and can be checked against incoming data. That said, these general sufficient conditions are far from necessary if we are given additional information about the return distributions.

\subsection{Economic limits to speculator behavior}

We now present some fundamental results that bound the speculator's decision-making. These results will be useful in developing the primary results of the paper in the next section. The next result introduces a lower bound to the speculator's STBL supply decision that arises from the fundamental price impact of repurchasing STBL.

\begin{proposition}\label{result:lt_lb}
Suppose the pre-decision collateral constraint is met at time $t$. There is a computable lower bound to $\Delta_t$.
\end{proposition}

We can interpret the lower bound in terms of a balance sheet constraint describing when the speculator's ETH position is exhausted in a repurchase. We give the specific bound in the proof but note that it is not especially useful on its own. Given information about the returns distribution and the level of current collateral and considering $\frac{\partial \psi}{\partial \mathcal L_t}$, much better bounds are possible. Note that if $\zeta >0$ is high enough, the lower bound may be the speculator's entire debt position, which would be expected in a liquid environment with heterogeneous agents.

\begin{center} \hyperlink{pf:lt_lb}{\texttt{[Link to Proof]}} \end{center}

The next result provides a useful upper bound to the speculator decision $\mathcal L_t$. The result is derived from incentives to issue STBL. Intuitively, it says that if supply is below this bound, then a speculator may see a profitable opportunity to expand supply. It's simply not profitable to issue more STBL than this bound. This doesn't mean that the speculator decides to achieve the bound, however, as it underestimates the liquidation costs that the speculator might face.\footnote{The model as formulated does not incorporate an interest rate paid by the speculator on issued STBL (the `stability fee' in Dai). Additionally, it does not incorporate a possible yield if the speculator creates STBL to lend on a lending platform as opposed to selling on the market. Under either of these extensions, Proposition~\ref{result:lt_ub} would change by an appropriate factor.} Notice that the bound is strongest when we have $\kappa \sim 1$.

\begin{proposition}\label{result:lt_ub}
Suppose either of the following hold for given $\kappa$:
\begin{itemize}
\item $\int_{\frac{c(L_t)}{X_t}}^{\frac{b(L_t)}{X_t}} \left(3-\frac{\alpha \mathcal D \bar N_t X_t z}{2(\bar N_tX_tz - \mathcal L_t)^2}\right) g_t(z) dz \leq 0$ and $\PR(A_t \cup B_t | \mathcal F_t) \geq \kappa^{-1} > 0$
\item $1 \geq \PR(A_t | \mathcal F_t) - 2\PR(B_t | \mathcal F_t) \geq \kappa^{-1} > 0$.
\end{itemize}
Then
$\mathcal L_t \leq \sqrt{\kappa \mathcal{L}_{t-1}\mathcal{D} \EX[X_{t+1}|\mathcal{F}_t]/X_t}$.
\end{proposition}

\begin{center} \hyperlink{pf:lt_ub}{\texttt{[Link to Proof]}}. \end{center}

The first condition comes from the derivative of the expected liquidation effect with respect to $\mathcal L_t$ taking $\beta=\frac{3}{2}$. The integrand can be interpreted as the effective leverage change in a given liquidation. This quantity is $<0$ evaluated at $b(L_t)$ (small liquidations effectively reduce leverage) whereas it is $>0$ evaluated at $c(L_t)$ (in very large liquidations, leverage reduction may not be effective due to effect on repurchase price). The integral condition then says that, in expectation, liquidations effectively reduce leverage. This is a  reasonable assumption given a starting state of sufficient over-collateralization, since reasonable distributions of $X_{t+1}$ will place most mass in the integral around $b(L_t)$ as opposed to $c(L_t)$, which is a tail event.

The second (alternative) condition says that the probability of having a liquidation is sufficiently smaller than not having a liquidation.

This result holds if \emph{either} of the two conditions hold, both of which could be checked in data-driven modeling. We will formalize an assumption like the first condition in the next section. Similar results going forward could be derived instead using a variation on the second condition.

%% file: content/5-domains.tex
\section{Stable and Unstable Domains}\label{sec:domains}

The primary results of the paper characterize regions in which the stablecoin price process can be interpreted as `stable' and `unstable'. In this section, we derive these results for the given model of a single speculator facing imperfectly elastic demand for STBL. In the next section, we consider generalizations of the model and how these results will differ given different design and market structures.

\subsection{Domain barriers/Stopped processes}

We first establish results in terms of barriers. While the stablecoin process is within certain barriers, we prove that it behaves in ways that are interpretable as `stable' and `unstable'. These barriers are generally stopping times, and we proceed by considering the stopped processes.

Assume that in the initial condition we have $\EX\left[ \frac{1}{\mathcal{L}_1} | \mathcal{F}_{0} \right] \leq \frac{1}{\mathcal{L}_{0}}$. We define the following stopping times:

\begin{itemize}
\item $\tau$ is the hitting time of $\EX\left[ \frac{1}{\mathcal{L}_{t+1}} | \mathcal{F}_{t}\right] > \frac{1}{\mathcal{L}_{t}}$

\item $T_m$ is the hitting time of $Z_t > m$, for $m \geq Z_0$

\item $S_1$ is the hitting time of $\EX[ \mathcal L_{t+1} | \mathcal F_t] < \mathcal L_t$

\item $S_2$ is the hitting time of $\EX[\mathcal L_{t+1} |\mathcal F_t] \geq \mathcal L_t$ such that $S_2>S_1$.
\end{itemize}

As we will see, while the stablecoin mechanism is working as intended, we generally expect the STBL supply to increase (equivalently in this setting, the STBL price to decrease, though in slow and bounded way). With this context in mind, $\tau$ represents the first time we \emph{expect} the STBL price to increase. Notice that this is an expectation of reciprocal of supply, a convex function, and so through Jensen's inequality, this is weaker than expecting the speculator to deleverage/reduce supply. 
 In particular, we have $\tau \leq S_1$.

Note that the expectations of the process are not necessarily the same as the  movements of the process: $\tau$ does not necessarily correspond to the first time the process actually increases in price. We track this with $T_m$, the time the STBL price breaches a given level above $Z_0$, which may be before or after $\tau$.

The stopping times $S_1$ and $S_2$ track when expectations about STBL supply change. These can be equivalently stated (and calculated in a data-driven model) based on expectations about the derivative of $\EX[Y_{t+2} | \mathcal F_t]$ with respect to $\mathcal L_{t+1}$ evaluated at $\mathcal L_t$, similarly to the discussion from the previous section on concavity.



Before proceeding, we formalize stopped versions of assumptions in Proposition~\ref{result:lt_ub}. The interpretation of these assumptions is the same as discussed in the previous section. Note that the results going forward could also apply more generally subject to additional stopping times embedding these assumptions. For notational simplicity, we just present the results subject to the stopping times already defined with the assumptions given.

\begin{assumption}
For $t\leq \tau$, $\PR(A_t \cup B_t |\mathcal F_t) = \PR(X_{t+1} \geq c(L_t)|\mathcal F_t) \geq \kappa^{-1} > 0$.
\end{assumption}

\begin{assumption}
For $t \leq \tau$, $\int_{\frac{c(L_t)}{X_t}}^{\frac{b(L_t)}{X_t}} \left(3-\frac{\alpha \mathcal D \bar N_t X_t z}{2(\bar N_t X_t z - \mathcal L_t)^2}\right) g_t(z) dz \leq 0$.
\end{assumption}

Notice that $\kappa$ will be $>1$ but $\sim 1$ as $X<c(L_t)$ is a low probability event.

Recall that the STBL price $Z_t$ is a function of collateral value, expectations about ETH returns, and expectations of liquidation costs (related to tail risks). These factors enter the speculator's supply decision, which then enters  $Z_t$. Going forward, we will explore how changes in these affect the STBL price process.

\subsection{`Stable' domain}

Subject to the barriers $\tau$ and $T_m$, the stablecoin process can be interpreted as stable in the following ways. In this domain, we derive bounds on large price movements and quadratic variation. We show below that for realistic values of parameters, the bounds are sufficiently powerful in practice.

Our first result bounds $Z_t$ under the condition $T_{Z_0} > \tau$. Conditioned on this, the price is contained within small variation--e.g., consider $Z_0=1$ and consider $\frac{1}{\kappa r} \sim 1$.
Recall that $r$  represents the upper bound on returns, $r = \sup_t \frac{\EX[X_{t+1} ]}{X_t}$, whereas $\kappa^{-1}$ is a lower bound for the probability that the collateral is not exhausted in a liquidation event,  $\PR(X_{t+1} \geq c(L_t)|\mathcal F_t)  \geq \kappa^{-1}$.

\begin{proposition}\label{result:stable_range}
If $T_{Z_0} > \tau$, then
$$Z_0 \geq Z_{t\wedge \tau} \geq \sqrt{\frac{\mathcal D}{\kappa \mathcal L_{t\wedge \tau -1} r}} \geq \frac{\mathcal{D}}{(\kappa \mathcal D r)^{\frac{2^t-1}{2^t}} \mathcal L_0^{\frac{1}{2^t}}}.$$
Furthermore for any $t$, $\mathcal L_{t \wedge \tau} \leq \kappa\mathcal D r$ and $Z_{t \wedge \tau} \geq \frac{1}{\kappa r}$.
\end{proposition}

\begin{center} \hyperlink{pf:stable_range}{\texttt{[Link to Proof]}} \end{center}

The condition $T_{Z_0} > \tau$ introduces dependence on future events. As such, we can't conclude with the information at time $t$ that the $t+1$ price is bounded in this way.


However, we can bound our expectations on the $t+1$ price given the information at time $t$ ($\mathcal F_t$). This approach relies on the fact that the versions of the process behave  as submartingales in the stopped setting.

\begin{proposition}\label{result:stable_submg}
$(\mathcal{L}_{t\wedge \tau})$ is a submartingale bounded above and $(Z_{t\wedge \tau})$ is a supermartingale bounded below. Thus they converge almost surely.
\end{proposition}

\begin{center} \hyperlink{pf:stable_submg}{\texttt{[Link to Proof]}} \end{center}

An immediate bound on expected price comes from the fact that stopped version of $Z_t$ is a supermartingale. This is the first result of the next proposition. Additionally, with a stronger assumption on $(X_t)$ that conditional expectation of returns is non-decreasing within the domain barriers, we can bound the expected price further.

\begin{proposition}\label{result:stable_exp_bound}
The process $(Z_{t\wedge\tau\wedge T_{Z_0}})$ is bounded in expectation by
$$Z_0 \geq \EX[ Z_{t\wedge \tau\wedge T_{Z_0}}] \geq \frac{1}{\kappa r}.$$
Further, assuming that for $t<\tau$, $(\EX[R_{t+1} | \mathcal F_t])$ is non-decreasing, then for $t\leq \tau$,
$$Z_{t-1} \geq \EX[Z_{t\wedge \tau} | \mathcal{F}_{t-1}] \geq \sqrt{ \frac{\mathcal{D}}{\kappa\mathcal{L}_{t-1}\EX[R_t |\mathcal{F}_{t-1}]}}.$$
\end{proposition}

\begin{center} \hyperlink{pf:stable_exp_bound}{\texttt{[Link to Proof]}} \end{center}

Going forward, we will work with a variation on the price process
$$Z_t^\prime := | m - Z_t| \hspace{1cm} \text{ for given } m\geq Z_0.$$
Using $m=1$, this has concrete interpretation as the absolute price deviation from the stablecoin peg. The stopped version of this process has the useful property of being a non-negative submartingale. In addition, $(Z_t^\prime)$ shares similar large deviation and quadratic variation properties with $(Z_t)$, which we explore in the remainder of this subsection. 

\begin{lemma}\label{result:zprime_submg}
The stopped process $(Z^\prime_{t\wedge \tau \wedge T_m})$ is a non-negative submartingale.
\end{lemma}

\begin{center} \hyperlink{pf:zprime_submg}{\texttt{[Link to Proof]}} \end{center}

We define the maximum process over some process $(\theta_t)$ as $\theta_N^* = \max_{t\leq N} |\Theta_t|$. The next result bounds the expected maximum of the deviation process $(Z_t)$.

\begin{proposition}\label{result:zprime_exp_max}
Suppose $m\geq Z_0$. Denote $E:= \EX[ Z_{\tau\wedge T_m} - m | Z_{\tau \wedge T_m} > m]$. Suppose any one of the following conditions holds:
\begin{itemize}
\item $\frac{1}{\kappa r} > m$ and $E > \frac{1}{\kappa r} - m$
\item $\frac{1}{\kappa r} = m$ and $E > 0$
\item $\frac{1}{\kappa r} < m$ and $E \geq 0$.
\end{itemize}
Then \hspace{3cm}
$\EX[ Z^{\prime *}_{\tau \wedge T_m}] \leq 2\left( m- \frac{1}{\kappa r}\right).$
\end{proposition}

\begin{center} \hyperlink{pf:zprime_exp_max}{\texttt{[Link to Proof]}} \end{center}

The value $(m-\frac{1}{\kappa r})$ describes the range of the domain considered. Prior to $T_m$, we know that the price falls in this range. The nontrivial part is describing what happens at the stopping time as it \emph{exceeds} this range if the stop is triggered by $T_m$. The value $E$ is the expected deviation at the stopping time \emph{given} that $T_m$ triggers the stop. By definition, we have that $E>0$. Given reasonable $\kappa$, $r$, and $m$, the condition for Proposition~\ref{result:zprime_exp_max} is satisfied quite broadly. For instance, the concrete instance with $m=1$ is satisfied since $\frac{1}{\kappa r} < 1$ taking into account the above discussion on $\kappa$.

Notice that the analysis for the proof can lead to better bounds if we have more information about $E$ or $p := \PR(Z_{\tau\wedge T_m} \leq m)$, e.g., by incorporating information from other results above or from knowledge about the distributions of $(X_t)$, such as from historical data. Additionally, the analysis can be used to bound either $E$ or $p$ given bounds on the other.

We now state the first main results of the paper. Our next result applies Doob's inequality to bound the probability of large deviations in the stopped process.

\begin{theorem}\label{result:zprime_doob_ineq}
For $m\geq Z_0$ and $\epsilon > 0$,
$$\PR\left(\max_{n \leq \tau \wedge T_m} Z_n^\prime > \epsilon \right)
\leq 2\epsilon^{-1} \left(m-\frac{1}{\kappa r}\right).$$
\end{theorem}

\begin{center} \hyperlink{pf:zprime_doob_ineq}{\texttt{[Link to Proof]}} \end{center}

The result can be quite powerful. Consider the concrete case of $m=1$, in which case $Z_t^\prime$ describes the deviation from the peg, and take (arguably reasonable) $\kappa^{-1} = 0.999$ ($99.9\%$ chance $X_t$ won't drop below $c(L_t)$) and $r$ annualized as 1.5 (daily $r=1.0011$). Then the probability that the stablecoin deviates from the peg by more than 0.1 is $\PR(Z^{\prime *}_{\tau\wedge T_1} > 0.1) \leq 0.042$.


Our next result derives from a form of Burkholder's inequality that applies to non-negative submartingales. We define the quadratic variation of $(Z^\prime_t)$ by
$$[Z^\prime]_t := \sum_{k=1}^t (Z^\prime_k - Z^\prime_{k-1})^2.$$
The quadratic variation is a stochastic process that measures how spread out the underlying process is. Its expectation at time $t$ is related to the variance at that time, supposing variance is defined--in particular, they are equal if the underlying process is a martingale. The result bounds the probability of large quadratic variation in the stopped process. In essence, with high probability, the quadratic variation can't be \emph{too far} away from the expected maximum.

\begin{theorem}\label{result:zprime_qv_pr}
Suppose $m \geq Z_0$ and $\epsilon > 0$. Then
$$\PR\left( \sqrt{[Z^\prime]_{\tau\wedge T_m}} > \epsilon \right) \leq 6 \epsilon^{-1} \left( m - \frac{1}{\kappa r}\right).$$
\end{theorem}

\begin{center} \hyperlink{pf:zprime_qv_pr}{\texttt{[Link to Proof]}} \end{center}

This result is also quite powerful. Considering the same setting as above, we have $$\PR( \sqrt{[Z^\prime]_{\tau\wedge T_1}} > 0.1) \leq 0.127$$ in the stable domain.

Bounds on the expectation of quadratic variation can also be obtained using a more classical form of Burkholder's inequality, albeit with stronger assumptions. We develop this idea in the next remark.

\begin{remark}
There is an additional form of Burkholder's inequality that extends to non-negative submartingales. If we are additionally given a useful bound on $\EX\Big[ \big(Z_{\tau\wedge T_m}^\prime\big)^p \Big]$ for some $1<p<\infty$ (for instance, if we have some distribution assumptions on $(X_t)$), then we can apply Lemma~3.1 in \cite{burkholder73} to derive the following bound on quadratic variation expectations:
$$\EX\Big[ \big([Z^\prime]_{\tau\wedge T_m}\big)^p \Big]^{\frac{1}{p}} \leq \frac{9p^{\frac{1}{2}}}{1-p^{-1}} \EX\Big[ \big(Z_{\tau\wedge T_m}^\prime\big)^p \Big]^{\frac{1}{p}}.$$
A topic of ongoing research is obtaining the Best constants/bounds in Burkholder's inequality, which may be able to tighten the bound.
The classical two-sided Burkholder inequailty may not extend to non-negative submartinagales. In general, only the first half of the Burkholder inequality (bounding expectations about quadratic variation by the maximum) extends to this setting and only for $1<p<\infty$. This contrasts with Proposition~\ref{result:zprime_qv_pr}, where we can derive results about probability of large quadratic variation of non-negative submartingales for the $p=1$ case. From a practical point of view, this may be sufficient.
\end{remark}

Notice that with an effective bound on the expectation of quadratic variation (QV) of the entire stable process, we have by law of large numbers
$$\frac{QV}{n} \rightarrow 0 \text{ as } n\rightarrow \infty.$$
So the longer the process is stable, the smaller the variability.

As we've characterized this `stable' domain based on $\tau$ and $T_m$, an exit from this region corresponds to either a change in expectations ($\tau$) or a large deviation event ($T_m$).
In actual applications, we will know when these stopping times arrive (or will at least have good measures of it, when hard to directly observe).
These could be used by system stakeholders as indicators that the local regime is changing.
Statistical analysis on historical data could also  predict how likely we are to see such indicators in coming steps.



\subsection{`Unstable' domain}

We now characterize how the stablecoin can be interpreted as unstable outside of the barriers described above. The intuition here is that the speculator's position is nearer to $c(L_t)$ and $b(L_t)$, and so expected costs of liquidation increase and are more sensitive to the threshold proximity, in addition to being driven by the volatile process $(X_t)$.
The remaining results in this section characterize a deflationary regime that is connected with instability in terms of forward-looking variance of stablecoin prices and large deviations. In this regime, we observe deleveraging spirals, which resemble short squeezes, and are counterintuitive as they lead to stablecoin price appreciation during times of collateral shock and lead to faster collateral drawdown.

Our next result characterizes a deflationary regime defined by stopping times $S_1$ and $S_2$. In such a setting, an opposite behavior occurs compared to the stable region: $(Z_t)$ behaves as a submartingale, tending to increase in price. The submartingale nature of the stablecoin price underpins the short squeezes within \emph{deleveraging spirals}.

\begin{theorem}\label{result:unstable_submg}
Restarting the process at $S_1$, we have that $(\mathcal L_{t\wedge S_2})$ is a supermartingale and $(Z_{t\wedge S_2})$ is a submartingale.
\end{theorem}

\begin{center} \hyperlink{pf:unstable_submg}{\texttt{[Link to Proof]}} \end{center}


The previous result guarantees that the process, after crossing $S_1$, enters a deflationary regime in a precise sense. This deflationary regime can be triggered by the factors affecting $S_1$, such as any of the following: shocks to collateral levels, increased expectations around deleveraging costs, or depressed ETH expectations. Similarly to the results above, in real applications, these stopping times can be used by stablecoin stakeholders as indicators that the local regime is changing and to statistically estimate the probable lengths of such deleveraging spirals.

The intuition behind deleveraging spirals is illustrated in Figure~\ref{fig:delev_spiral_illustrate}. In an equilibrium, the stablecoin supply is matched to demand. As a first wave of speculator liquidations occur, whether voluntary deleveraging or automated by the protocol, collateral is used to repurchase the stablecoin to reduce the supply. In an imperfectly elastic market, this causes an imbalance in demand relative to supply, and an increase in stablecoin price is needed to reduce demand. This has an amplifying effect, however, in follow-on rounds of liquidations: more collateral is needed to reduce supply by the same amount because of the increased stablecoin price, and each round of liquidations continues to increase the stablecoin price.

\begin{figure}
	\centering
	\includegraphics[width=0.8\textwidth]{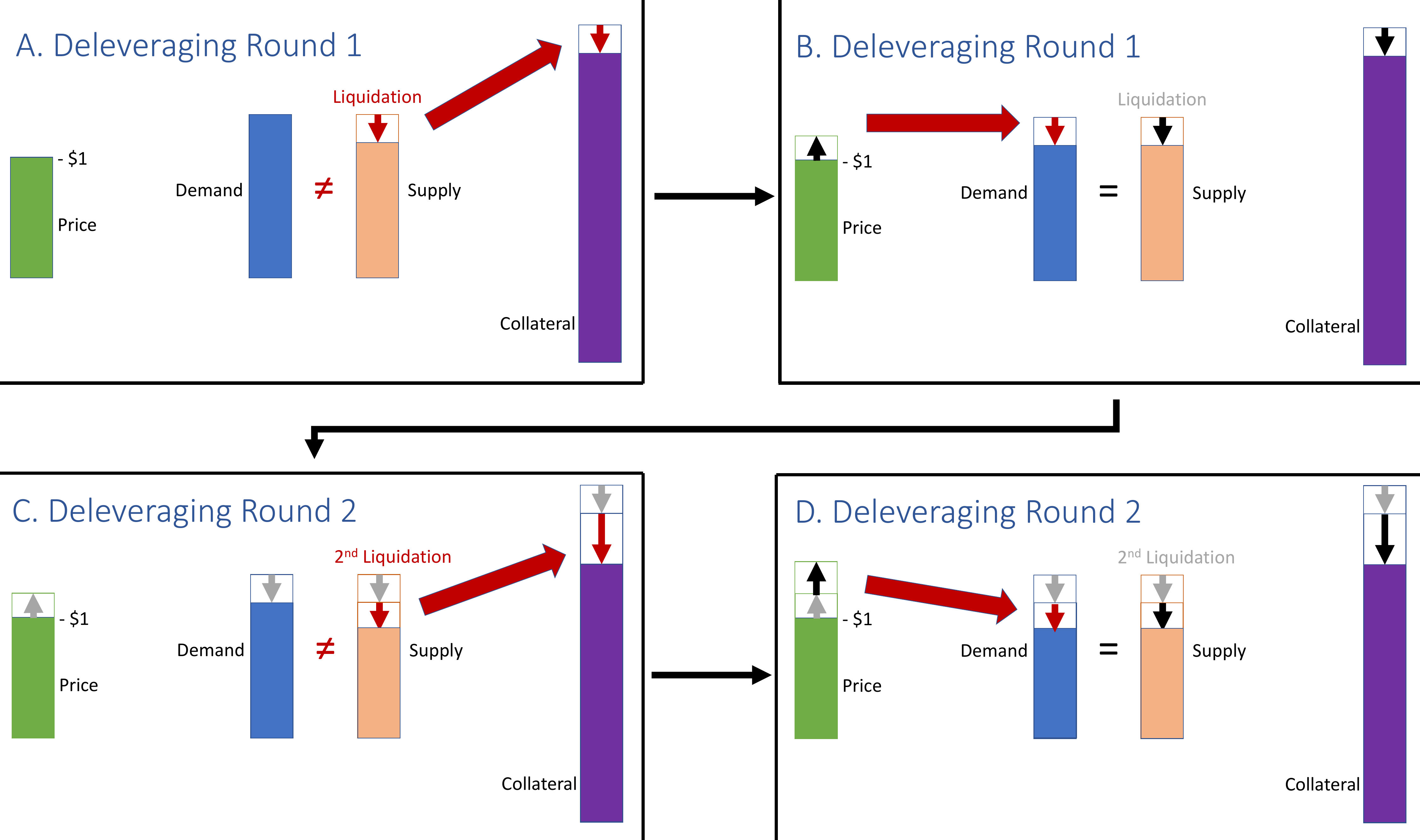}
	\caption{Illustration of deleveraging spirals. In liquidations, collateral is used to reduce supply. Stablecoin price rises in response to imbalance with demand. This has an amplifying effect in follow-on liquidations.}\label{fig:delev_spiral_illustrate}
\end{figure}

Black Thursday in March 2020 provides strong evidence of deleveraging spirals in the Dai stablecoin. ETH price crashed $\sim 50\%$ on 12 March 2020 (Figure~\ref{fig:eth_mar20}) This triggered a wave of liquidations in Dai, as well as other cryptocurrency systems. These liquidations led to a cornering effect from deleveraging spirals in the Dai market, as shown in Figure~\ref{fig:dai_mar20}.  Speculators faced premiums in excess of 10\% to deleverage during the crisis and lingering premiums $>2\%$ several weeks after. The cornering effect is also supported by lending rates on Dai, which reached high double digits during the crisis (Figure~\ref{fig:dai_lending_mar20}).
Maker was also affected by global mempool flooding on Ethereum during the crisis, which caused many Dai liquidation auctions to clear at near zero prices. This had the effect of amplifying the deleveraging effect on collateral and led to a \$4m shortfall in the system. See \cite{blocknative2020,maker_shortfall2020} for more details. Many market participants were surprised in this crisis that Dai traded at significant premiums despite the much riskier state of Maker in terms of collateral and liquidations, which our model explains as deleveraging spirals.

\begin{figure}
	\centering
	\begin{subfigure}[b]{0.4\textwidth}
		\includegraphics[width=\textwidth]{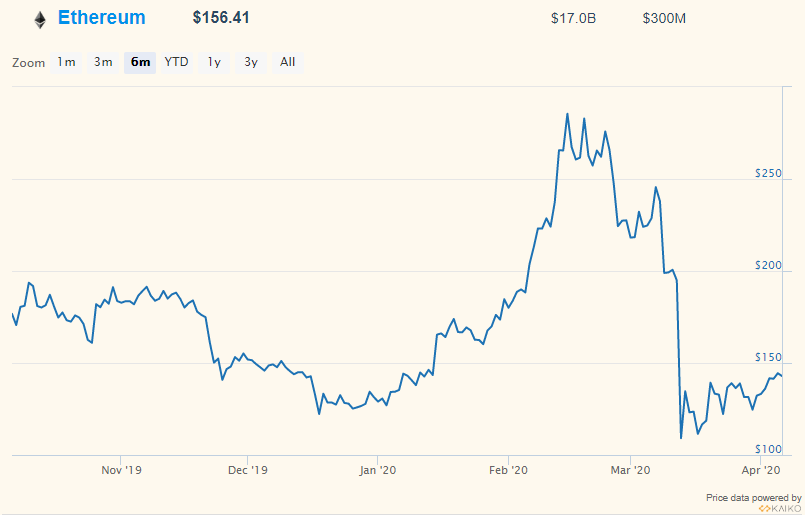}
		\caption{}\label{fig:eth_mar20}
	\end{subfigure}
	\begin{subfigure}[b]{0.4\textwidth}
		\includegraphics[width=\textwidth]{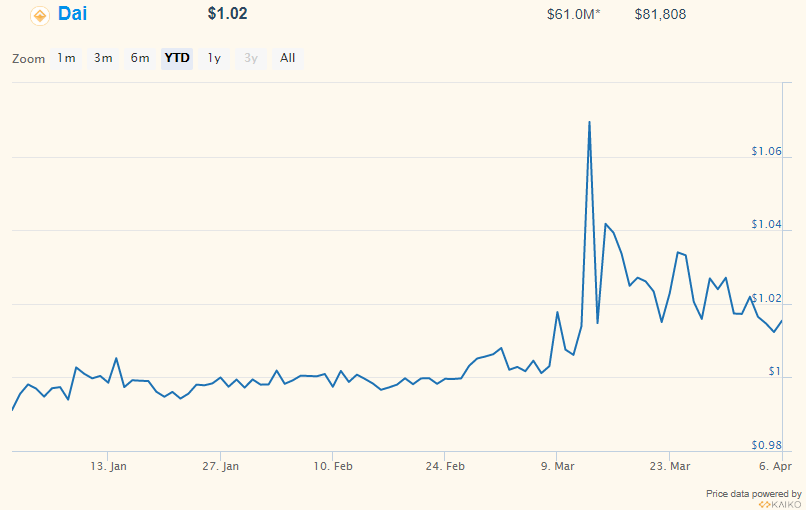}
		\caption{}\label{fig:dai_mar20}
	\end{subfigure}
	\begin{subfigure}[b]{0.5\textwidth}
		\includegraphics[width=\textwidth]{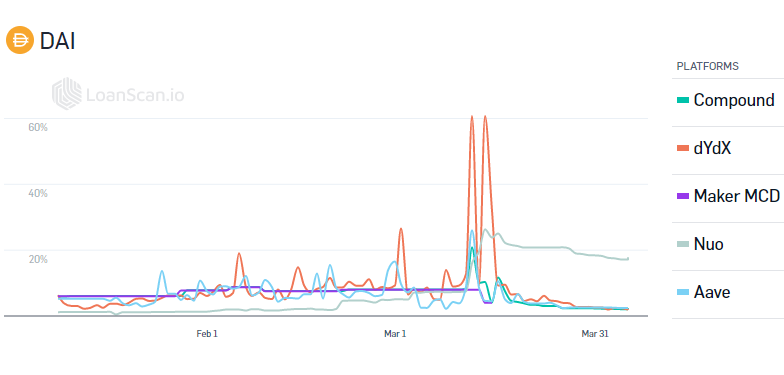}
		\caption{}\label{fig:dai_lending_mar20}
	\end{subfigure}
	\caption{Black Thursday in March 2020. (a) $\sim 50\%$ ETH price crash (OnChainFX). (b) Deleveraging effects on Dai price and volatility (OnChainFX). (c) Deleveraging effects on Dai lending rate (LoanScan)}\label{fig:mar20}
\end{figure}

\begin{remark}
(Interaction with cascading liquidations)
A different type of deleveraging spiral can occur in debt security models when the collateral asset price is endogenous to the model and can be depressed from the market impact of liquidations (e.g., fire sales). In this context, liquidations can cascade with a first round of liquidations triggering a follow-up rounds due to the impact on the collateral market. Conceptually, when this endogenous collateral effect is added to our model, the two deleveraging spiral types amplify each other. In particular, when the price of the stablecoin increases from the effects described above, more collateral must be liquidated to deleverage the same amount, and this greater collateral liquidation has a higher impact on the collateral asset market, which can trigger further liquidations cyclically in larger size than with the fire sale effect solely. We discuss how to endogenize collateral asset prices to the model in the Appendix.
\end{remark}

We now derive practical tools that will connect these regimes containing deleveraging spirals with instability in terms of forward-looking price variance of the stablecoin, and which do not require the detection of whether $S_1$ has occurred. This formalizes the high price variation observed in Dai during and after Black Thursday. We begin in the next remark by setting up a variance estimation idea based on Taylor approximation.

\begin{remark}\label{remark:var_approx}
(Estimating variances) Taylor approximations can be applied to estimate the variances of the stablecoin process. Consider $X_t = X_{t-1} R_t$ for return $R_t \geq 0$. For notational clarity, define\footnote{As in the case of $\psi$, $h$ could have a subscript $t$ (or equivalently other time $t$ inputs), but we relax notation as we only use in the context of time $t$.}
$$h(\rho, n) := \arg\max_{\mathcal{L}_t} \EX[Y_{t+1} | \mathcal F_t] = \mathcal L_t,$$
where $\rho,n$ are realizations of $R_t,\bar N_t$. Variance in stablecoin supply follows
$$\text{Var}(\mathcal{L}_t | \mathcal F_{t-1}) \approx h^\prime \left(\EX[R_t | \mathcal F_{t-1}], \bar N_t \right)^2 \text{Var}(R_t | \mathcal F_{t-1})$$
and the stablecoin price \textbf{variance approximation} is
\begin{equation}\label{eqn:var_approx}
\text{Var}(Z_t | \mathcal F_{t-1}) \approx \frac{\mathcal D h^\prime (\EX[R_t | \mathcal F_{t-1}], \bar N_t )^2}{\EX [\mathcal L_t | \mathcal F_{t-1}]^4} \text{Var}(R_t | \mathcal F_{t-1}).
\end{equation}
This is given informally, but could in principle be formalized using two steps of compounded Taylor approximation error. The approximation error is arguably moderate considering that our domain is bounded away from singularities (e.g., our lower bound results on $\mathcal L$).
\end{remark}

This variance approximation (Eq.~\ref{eqn:var_approx} in Remark~\ref{remark:var_approx}) is low in the stable domain and can be high in the unstable domain, as formalized in the following Theorem~\ref{result:var_approx}. We introduce a few more assumptions that we use only in deriving the remaining results in this section. All of these assumptions come down to assumed properties of the $R_t$ distribution.


\begin{assumption}
The post-decision collateral constraint at time $t$ is not binding in the speculator's maximization.
\end{assumption}

This first assumption means that the speculator's objective fully accounts for the post-decision collateral constraint (i.e., by maximizing the objective, the speculator by extension also satisfies the constraint). This is reasonable unless expected returns are excessively high.

\begin{assumption}
Returns $R_{t-1}$ and $R_t$ are independent.
\end{assumption}

\begin{assumption}
$\psi$ is twice continuously differentiable.
\end{assumption}

This last assumption restricts the density $g_t$. We now present the result, which applies the implicit function theorem to derive the derivatives of $h$, which describe the sensitivity of $h$ to price and collateral level.

\begin{theorem}\label{result:var_approx}
Under the above assumptions, the following hold:
\begin{enumerate}
\item $\frac{\partial}{\partial \rho} h(\rho,n)$ $\frac{\partial}{\partial n} h(\rho,n)$ exist;
\item $\frac{\partial}{\partial \rho} h(\rho,n) \geq 0$ and is increasing in $-\rho$ by order of $\frac{1}{\rho}$ for $\rho\geq \frac{b_{t-1}}{X_{t-1}}$, $\mathcal L_t > 8$;
\item $\frac{\partial}{\partial n} h(\rho,n) \geq 0$ and is increasing in $-n$ by order of $\frac{1}{n}$ for $n \geq \frac{b_{t-1}}{X_t}$, $\mathcal L_t > 8$;
\item $\exists \varepsilon$ with $0< \varepsilon < 1$, s.t. $\frac{\partial}{\partial \rho} h(\rho,n) > 1$ if $\rho <\varepsilon$ , $L_t > \frac{27}{46}\alpha\mathcal D$, and $c_t>2$.
\end{enumerate}
As a result, the variance approximation in Eq.~\ref{eqn:var_approx} increases by order of $\frac{1}{R_t^2}$ in $-R_t$ and $\frac{1}{\bar N_t^2}$ in $-\bar N_t$.
\end{theorem}

\begin{center} \hyperlink{pf:var_approx}{\texttt{[Link to Proof]}} \end{center}

Theorem~\ref{result:var_approx} shows that the variance approximation in Eq.~\ref{eqn:var_approx} in Remark~\ref{remark:var_approx} increases by order of $\frac{1}{R_t^2}$ during an ETH return shock (result 2). Recall that $R_t$ is multiplicative return, and so the effect is large for a significant shock $R_t<1$. Similarly, settings with lower collateralization in the initial conditions have higher variance approximation by order of $\frac{1}{\bar N_t^2}$ (result 3). Such differences in initial conditions of collateral could result from, for example, different realizations of liquidations or the speculator abandoning its collateral position (and so extracting any excess collateral it can). Result 4 shows that there are cases where the $h^\prime$ factor in the variance approximation is $> 1$, meaning that the variance of $R_t$, the inherently volatile process, will carry through directly to $Z_t$, the `stable' process.

Note that the extra conditions on the scale of $\mathcal L_t$ and $c_t$ in Theorem~\ref{result:var_approx} results 2-4 may seem strange at first sight. Since the $(Z_t)$ process is scale-invariant, as proven in Proposition~\ref{result:rescalings}, the results about $Z_t$ variance hold more generally. In particular, recall that a term of $\sim \frac{1}{\mathcal L_t}$ shows up in the variance approximation in Remark~\ref{remark:var_approx}, which will cancel out the conditions on scale.


Up to this point, we have only been able to say things about variance estimations. We will now show that the `stable' and `unstable' regimes are well-interpreted in the following sense: given different initial conditions of the same process, the forward-looking stablecoin price variances are indeed distinct. If we start in the unstable regime, we will always have higher variance than if we start in the stable regime. The next result formalizes this.

\begin{theorem}\label{result:unstable_var}
In addition to the previous assumptions, suppose $X_t \geq b(L_{t-1}) + \epsilon$ for some $\epsilon>0$ (the pre-decision collateral constraint is exceeded by $\epsilon$, which restricts the ranges of both $X_t$ and $\bar N_{t-1}$).
Consider two possible states $s$ and $u$ of the stablecoin at time $t$ that differ only in collateral amounts $\bar N_{t-1}^s > N_{t-1}^u$ and evolve driven by the common price process $(X_t)$. Then the forward-looking price variances satisfy
$$\text{Var}(Z_t^s | \mathcal F_{t-1}) < \text{Var}(Z_t^u | \mathcal F_{t-1}).$$
\end{theorem}

\begin{center} \hyperlink{pf:unstable_var}{\texttt{[Link to Proof]}} \end{center}

Special care should be given to the treatment of $Z_t$ under the condition $X_t \leq c(L_{t-1})$, as the STBL price may no longer be well-defined without $\zeta>0$ as no collateral remains. In a real system, this is equivalent to the event that all speculators are wiped out. The reason for our condition on $X_t$ in the above result is partly to keep things well-defined and partly because there can be a non-smooth point in $h$ at $X_t = b(L_{t-1})$.

Similar variance difference results can be derived for varying initial conditions of $X_{t-1}$ and $\mathcal L_{t-1}$ as opposed to $\bar N_{t-1}$. In some sense, these are all similar as they change the initial collateralization level, though there will be some difference in price effect.


These analytical results describe regimes in which the stablecoin can be interpreted as stable and unstable. As we have discussed, they can be adapted into data-driven risk tools, for instance to estimate probabilities of peg deviations and to infer about how likely regimes are to change in the near future.

While these results apply over limited steps ahead--e.g., forward-looking variance is derived for the next time period--they \emph{do} point in the right direction that stability domains are related to traditional measures in finance. Naturally, it would be good to have results describing further periods into the future. In principle, these could be estimated, although the process in this section is already complex. The fact that we are able to relate these regimes analytically to forward-looking variance is already a step ahead, and a valuable new result in its own right. We conjecture that it could work similarly over multi-steps, though in less tractable ways.

%% file: content/6-extensions.tex
\section{Stability in `Perfect' Settings}\label{sec:perfect_stability}

In the previous section, we considered the given model of a single speculator facing imperfectly elastic demand for STBL. We now consider idealized settings, in which STBL demand is perfectly elastic and/or unlimited speculator supply exists. In these idealized settings, we demonstrate that stablecoin can be interpreted as well-stabilized.

\subsection{Perfectly elastic demand}
Under perfectly elastic demand, STBL demand is time-dependent $\mathcal D_t$, which adapts in each time period to match STBL supply. This results in $Z_t = 1$. In this case, the speculator's issue and repurchase price is always \$1 and $\$\alpha$ in a liquidation. The problem simplifies to evaluating

$$\begin{aligned}
\EX[Y_{t+1} | \mathcal F_t]
&= \Delta_t\EX[R_{t+1} | \mathcal F_t]
+ \int_{\frac{c_t}{X_t}}^\infty (\bar N_t X_t z - \mathcal L_t)g(z)dz \\
&\hspace{0.5cm}
+ (1-\alpha)\int_{\frac{c_t}{X_t}}^{\frac{b_t}{X_t}} \frac{\beta L_t - \bar N_t X_{t}z}{\beta-1}  g(z)dz,
\end{aligned}$$
where the liquidation effect is now $\ell_{t+1} (1-\alpha)$ where $\ell_{t+1} = \frac{\beta L_t - \bar N_t X_{t+1}}{\beta-1}$.

In this setting, we have
$\frac{\partial \psi}{\partial \mathcal L_t} = \EX[R_{t+1}|\mathcal F_t] - \PR(A_t\cup B_t) - \frac{\beta}{\beta - 1}(\alpha-1) \PR(B_t)$.
Recalling that $\PR(A_t)$ and $\PR(B_t)$ are functions of $\mathcal L_t$ and supposing a non-binding collateral constraint, the speculator chooses $\mathcal L_t$ such that
$$\EX[R_{t+1} |\mathcal F_t] = \PR(A_t \cup B_t) + \frac{\beta}{\beta - 1}(\alpha-1) \PR(B_t).$$
Noting that $\EX[R_{t+1}] \geq 1$, $\PR(A_t \cup B_t)$ is decreasing in $\mathcal L_t$ but generally $\sim 1$, and $\PR(B_t)$ is increasing in $\mathcal L_t$, this is interpretable as the speculator balancing expected return against $\frac{\beta}{\beta-1}\times$  the expected (constant) liquidation cost in deciding whether to issue a new unit of STBL.

In this setting, the STBL price is identically \$1 and the speculator only faces the risk of leveraged ETH declines subject to a fixed liquidation fee. Liquidations generally work well to keep the system over-collateralized, and the only real risk to STBL holders is from extreme single period declines in ETH price.


\subsection{Unlimited speculator capital supply}
Suppose there is an infinite depth of speculator's capital ready to enter the STBL market given what they see as a profitable opportunity subject to STBL demand $\mathcal D$. 
The speculator in such a market would choose to deposit collateral and issue new STBL at time $t$ if $\frac{\mathcal D \mathcal L_{t-1}}{\mathcal L_t^2} \EX[R_{t+1} | \mathcal F_t]  - \gamma \geq 0$, where $\gamma$ represents the representative speculator's expected liability and liquidation cost after entering the market. Arguably, $\gamma \sim 1$ as, in an infinite depth market, the  speculator can start from a position of low leverage.

The speculator's profitability (for the marginal  STBL issue) will be 0, which yields equality in the above condition, and therefore,
$$\mathcal L_t = \sqrt{\gamma \mathcal D \mathcal L_{t-1} \EX[R_{t+1} | \mathcal F_t]}.$$
Notice the similarity with the upper bound in Proposition ~\ref{result:lt_ub}. 

Further using that $(X_t)$ is a submartingale, in which case $\EX[R_{t+1}|\mathcal F_t] \geq 1$, we find the  STBL price is constrained to a small range of $Z_0 \geq Z_t \geq \frac{1}{\gamma r}$. This resembles the perfectly elastic demand case. In this case speculators are able to liquidate positions without influencing STBL price, while in the infinite depth case because the speculator is always willing to issue new STBL to offset a liquidation.

\subsection{No stable region if $(X_t)$ is not a submartingale}
The mechanisms that make the idealized settings well-stabilized break down when the ETH price process $(X_t)$ is not a submartingale. This stresses how fragile the stablecoin market is to negative expectations in the primary ETH market, even under these idealized settings. In the unlimited speculator case,  speculators no longer enter the market if expectations are negative, and so we don't achieve the supply bound developed above. Instead, we return to the main setting of the paper, which can be interpreted as unstable under negative expectations as it leads to deleveraging effects. In the perfectly elastic demand setting, the STBL supply goes to zero as the speculator chooses not to participate.

%% file: content/7-discussion.tex
\section{Discussion}\label{sec:discussion}

This paper presents a new stochastic model of non-custodial over-collateralized stablecoins,  where the collateral has value exogenous to the stablecoin system and the stablecoin has an endogenous market price. These stablecoins bear a resemblance to a non-custodial form of the current monetary system of commercial bank money but give rise to new risks such as those experienced on Black Thursday. These stablecoins stand in contrast to unbacked or endogenously backed stablecoins, such as Terra UST, which are better understood using tools of insolvency and currency peg models, as well as custodial stablecoins such as Tether, which can resemble the underlying structures of narrow banks or money market funds.

In our model, we formally characterize domains that can be interpreted as stable and unstable for the stablecoin. By bounding the probability of large deviations and the quadratic variation of the price process, we prove that the stablecoin behaves in a stable way when restricted to a certain region. In contrast, price variance is shown to be distinctly greater in a separate region. This is triggered by large deviations, collapsed expectations, and liquidity problems from deleveraging. We also characterize a deflationary deleveraging spiral as a submartingale, which can exacerbate liquidity problems in a crisis. These deleveraging spirals resemble short squeezes, and are counterintuitive as they lead to stablecoin price appreciation during times of shock, whereas we might otherwise expect prices to depreciate given the riskier state of the system. Further, this appreciation is detrimental: it leads to faster collateral drawdown, and potentially shortfalls, as more collateral is required to fulfill liquidations and is accompanied by higher price variance.

An observation from the model is that the speculator chooses a collateral level \emph{above} the required collateral factor. This is because the expected liquidation cost is greater than the \$1 face value. The speculator will desire to increase the collateralization during times when the expected liquidation cost is higher, which can occur after a shock to collateral value or if the speculator expects the collateral to be more volatile. This generally explains the high level of over-collateralization seen in Dai, which typically ranges $2.5-5\times$ although the collateral factor is $1.5\times$.

The presence of deleveraging effects poses fundamental trade-offs in decentralized design. One way to bring the stablecoin closer to the `perfect' stability cases is to increase elasticity of demand. This relies on the presence of good uncorrelated alternatives to the stablecoin. As all non-custodial stablecoins likely face similar deleveraging risks, greater elasticity relies on custodial stablecoins or greater exchangeability to fiat currencies. Another way to bring the stablecoin closer `perfect' stability is to increase the supply of new speculators. As there will not be unlimited supply of speculators with positive ETH expectations (especially during an extended bear market), this relies on having another uncorrelated collateral asset. As all decentralized assets are very correlated, this again largely relies on including custodial collateral assets, like Maker's recent addition of USDC.\footnote{Recall that custodial assets face their own risks, however, which may not be uncorrelated in extreme crises.  Custodial stablecoins  are subject to counterparty risk,  systematic
risks, bank run risks, asset seizure risk, and effects from negative interest rates. The treasury secretary J. Yellen referred to the materialization of these  risks in her annual testimony in front of the Senate Banking Committee, on May 10th 2022:  “A stablecoin known as TerraUSD experienced a run and declined in value,” Yellen said. “I think that this simply illustrates that this is a rapidly growing product and there are rapidly growing risks.”} While these measures strengthen the stability results, it's at the expense of greater centralization and moves the system away from being `non-custodial'.

We suggest a way to improve the design of Dai's savings pool toward damping deleveraging effects without greater centralization through incentivizing exchangeability of Dai during deleveraging events. In its current state, the Maker system charges fees to speculators, part of which it passes on to Dai holders as an interest rate if the holder locks the Dai into a savings pool. With modified mechanics, this savings pool can provide a buffer to deleveraging effects. For instance, if we allow Dai in the savings pool to be bought out at a reasonable premium to face value by a speculator who uses it to deleverage, then deleveraging effects are bounded by the premium amount up to the size of the savings buffer. The Dai holders who participate in this savings pool are then compensated for providing a repurchase option to the speculator. The Dai holder could elect to have the repurchase fulfilled in the collateral asset, or something else, like a custodial stablecoin. In this way, this mechanism can provide some of the benefits of the `perfect' stability settings while enabling Dai holders to choose how decentralized they want to be. A Dai holder who does not require high decentralization would elect to receive the compensation from the savings pool whereas a Dai holder who requires higher decentralization would choose not to use the savings pool. Our model can be extended to consider such mechanisms.

Since the release of our paper, mechanisms resembling this, which try to boost liquidity around liquidations to quell deleveraging spirals, have been adopted by projects such as  \cite{liquity}. Empirically, these mechanisms have the effect of smoothing deleveraging effects over a longer time period, lowering the effect of shocks but not entirely removing the short squeeze effect (see Figure~\ref{fig:liquity-stability-pool}).
Maker has chosen to go a different direction by maintaining direct exchangeability with the custodial USDC \cite{maker_psm}, which has allowed Dai to maintain a close peg through subsequent crises at the expense of heavy reliance on custodial stablecoins. The stablecoin Rai has chosen a third path of instituting negative rates on stablecoin holders during crises \cite{rai} via a PID controller, which is effectively charging stablecoin holders insurance premiums when demand for stablecoins outweights demand for leverage, thus lowering demand to help attain peg.

\begin{figure}
	\centering
	\includegraphics[width=0.5\textwidth]{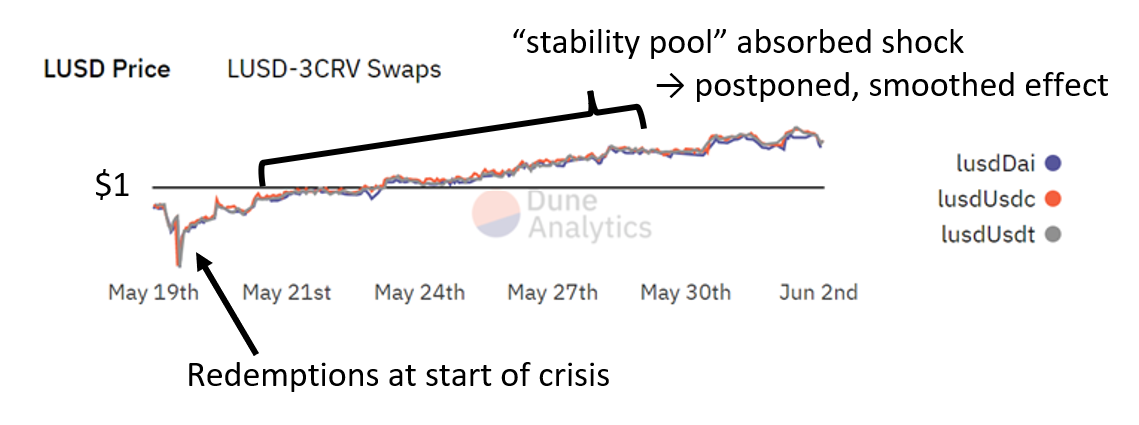}
	\caption{Effect of Liquity's stability pool on LUSD price in Curve's on-chain market in the May 2021 crisis. Deleveraging effect is delayed and smoother compared to Dai's price effect on Black Thursday (cf. Figure~\ref{fig:dai_mar20}).}\label{fig:liquity-stability-pool}
\end{figure}

Our model and results can also apply more broadly to synthetic and cross-chain assets and over-collateralized lending protocols that allow borrowing of illiquid and/or inelastic assets-- whenever the mechanism is based on leveraged positions and leads to an endogenous price of the created or borrowed asset. 
We have characterized the risk that such structures feature intertwining of collateral liquidation spirals and short squeezes of the created asset.
Synthetic assets generally use a similar mechanism just with a different target peg. Cross-chain assets that port an asset from a blockchain without smart contract capability (e.g., Bitcoin) to a blockchain with smart contracts (e.g., Ethereum) also tend to rely on a similar mechanism. In non-custodial constructions such as \cite{zamyatin19} and \cite{tbtc}, vault operators are required to lock ETH collateral in addition to the deliverable BTC asset. They bear a leveraged ETH/BTC exchange rate risk and face similar deleveraging risk. In particular, to reduce exposure, they need to repurchase the version of the cross-chain asset on Ethereum.

Several generalizations of analytical results are left for future research.  Here we considered  collateral prices exogenous, but it would be interesting to model market impact effects of large collateral liquidations and also enable modeling of stablecoins like Synthetix sUSD that have \emph{endogenous} collateral (see \cite{klagesmundt2020stablecoins}). One possible way to endogenize collateral prices is via an inverse demand function.
We expect that the general methods used in this paper can be applied to partial equilibrium settings such as this. Naturally, this would necessitate conditions on the inverse demand function that ensure that the expected returns as a function of  the issuance remains concave.




We have specified the speculator's decision-making in terms of a sequence of one-period optimization problems. Alternatively, the speculator could strategically coordinate the sequence of decisions further into the future and develop long-term strategies. This could be formulated by using an exit time for the speculator, when they can cash our their position by selling to someone else at par.  If this terminal time is deterministic, the problem can be formulated as a dynamic program, in which the terminal decision results from the one-period optimization, intermediate decisions solve a Bellman equation conditioned on the information revealed up to that point, and random returns are independent.  For instance, \cite{biais18} sets up a supermodular game in a setting where agents exit at a random exponential time.
 in t
The model could be extended to  include multiple speculators. Speculators have in reality a finite depth and moreover, they maintain positions with different leverage points and ETH expectations. This can lead to a sequential schedule of liquidation points at a given time throughout the system, which will be reflected in a speculator's expected liquidation costs. A given speculator will take into account price effects from the potential liquidations of other speculators' positions in addition to their own, see \cite{mincawissel} for leveraging-deleveraging games in the traditional banking system. Here, the speculator's value depends on liquidation costs  and on the supply limit imposed by the finite market depth.  Incorporating strategic aspects is left for future research.

%

%% file: content/9b-proofs.tex
\section{Proofs}\label{sec:proofs}

In the proofs, we often use the following elementary result
\begin{lemma}\label{prop:elem_ineq}
For $\alpha,\mathcal D, L \geq 0$,
$$\alpha\mathcal D + L \leq \sqrt{\alpha^2\mathcal D^2 + 4\alpha\mathcal D L +L^2} \leq \min\Big(2\alpha \mathcal D + L, \alpha \mathcal D + L + \sqrt{2\alpha \mathcal D L}\Big).$$
\end{lemma}

\begin{proof}
Define $\varepsilon := \sqrt{\alpha^2\mathcal D^2 + 4\alpha\mathcal D L +L^2}$. We have $\varepsilon \leq 2\alpha \mathcal D + L$ as long as $\alpha \mathcal D \geq L(\sqrt{3}-2)$, which is true since $\alpha,\mathcal D,L\geq 0$. Next, notice that $\varepsilon = \sqrt{(\alpha\mathcal D + L)^2 + 2\alpha\mathcal D L}$. Thus $\varepsilon > \alpha\mathcal D + L$ since $2\alpha\mathcal DL \geq 0$. Lastly, by concavity, $\varepsilon \leq \alpha\mathcal D + L + \sqrt{2\alpha \mathcal D L}$.
\end{proof}

\noindent\rule{\textwidth}{1pt}
\paragraph{Proposition~\ref{result:exy_concave}} \hypertarget{pf:exy_concave}{}
\begin{proof}
Consider $X_{t+1} = X_t R_{t+1}$. For notational simplicity, drop subscripts as follows: $\bar N_t \mapsto N$, $X_{t} \mapsto X$, $\mathcal L_t \mapsto \mathcal L$, $\Delta = \mathcal L_t - \mathcal L_{t-1}$, $c(L_t) \mapsto c$, $b(L_t) \mapsto b$, $g_t \mapsto g$, $R_{t+1} \mapsto R$. Define $\psi := \EX[Y_{t+1} | \mathcal F_t]$. Then
$$\psi(\mathcal L)
= \frac{\Delta \cdot \mathcal D}{\mathcal L}\EX[R | \mathcal F_t]
+ \int_{c/X}^\infty (NXz - \mathcal L)g(z)dz
+ \int_{c/X}^{b/X} \left(3\mathcal L - \frac{\alpha\mathcal D \mathcal L}{2NXz - \mathcal L} - 2NXz\right)g(z)dz.$$
Recall that the integrand factor $\left(3\mathcal L - \frac{\alpha\mathcal D \mathcal L}{2NXz - \mathcal L} - 2NXz\right)$ evaluated at $Xz=c$ is $\mathcal L - Nc$ (the liquidation zeros out the speculator's collateral position), and evaluated at $Xz=b$ is 0 (on the threshold of liquidation).

We obtain
$$\begin{aligned}
\frac{\partial \psi}{\partial \mathcal L}
&= \frac{\mathcal D \mathcal L_{t-1}}{\mathcal L^2} \EX[R | \mathcal F_t]
- \left(NX \frac{c}{X} - \mathcal L\right) g\left(\frac{c}{X}\right) \frac{\partial c}{\partial \mathcal L} \frac{1}{X} - \int_{\frac{c}{X}}^\infty g(z) dz \\
&\hspace{0.5cm}- \left(\mathcal L - NX\frac{c}{X}\right) g\left(\frac{c}{X}\right) \frac{\partial c}{\partial \mathcal L} \frac{1}{X} + \int_{\frac{c}{X}}^{\frac{b}{X}} \left(3-\frac{\alpha \mathcal D N X z}{2(NXz - \mathcal L)^2}\right) g(z) dz \\
&= \frac{\mathcal D \mathcal L_{t-1}}{\mathcal L^2} \EX[R | \mathcal F_t]
- \int_{\frac{c}{X}}^\infty g(z) dz
+ \int_{\frac{c}{X}}^{\frac{b}{X}} \left(3-\frac{\alpha \mathcal D N X z}{2(NXz - \mathcal L)^2}\right) g(z) dz 
\end{aligned}$$

$$\begin{aligned}
\frac{\partial^2 \psi}{\partial \mathcal L^2}
&= -\frac{2\mathcal D \mathcal L_{t-1}}{\mathcal L^3}  \EX[R | \mathcal F_t]
+ g\left( \frac{b}{X}\right) \frac{\partial b}{\partial L}\frac{1}{X} \left(3-\frac{\alpha \mathcal D N b}{2(Nb - \mathcal L)^2}\right) \\
&\hspace{0.5cm} - g\left( \frac{c}{X}\right) \frac{\partial c}{\partial L}\frac{1}{X} \left(2-\frac{\alpha \mathcal D N c}{2(Nc - \mathcal L)^2}\right)
- \int_{\frac{c}{X}}^{\frac{b}{X}} \frac{\alpha \mathcal D N X z}{(NXz-\mathcal L)^3}g(z)dz.
\end{aligned}$$

Notice that $\frac{\partial b}{\partial L} > 0$, $\frac{\partial c}{\partial L} > 0$, $g \geq 0$, and
$$3-\frac{\alpha \mathcal D N b}{2(Nb - \mathcal L)^2} = 3- \frac{\alpha \mathcal D \beta \mathcal L}{2(\mathcal L (\beta-1))^2} = 3 - \frac{3\alpha \mathcal D}{\mathcal L} < 0$$
by assumption that liquidation repurchase price always $\geq 1$. Additionally, the remaining integral is always positive as the integrand is positive between the limits  and $g \geq 0$. Finally, $\EX[R|\mathcal F_t] \geq 0$ since $(X_t)$ is a submartingale. Thus under the given conditions, $\frac{\partial^2\psi}{\partial\mathcal L^2} \leq 0$ as all terms are $\leq 0$.

Further supposing that either $\EX[R|\mathcal F_t] > 0$ or $\PR\Big(c(L) < XR < b(L)\Big) = \int_{c/X}^{b/X} g(z) dz > 0$, then $\frac{\partial^2 \psi}{\partial \mathcal L^2} < 0$.
\end{proof}

Notice that the $\frac{1}{2}$ in the bound is related to the choice $\beta=\frac{3}{2}$.

\noindent\rule{\textwidth}{1pt}
\paragraph{Proposition~\ref{result:rescalings}} \hypertarget{pf:rescalings}{}
\begin{proof}
Easily verifiable by substitution, noting that factors of $\gamma$ cancel in the integral limits.
\end{proof}

\noindent\rule{\textwidth}{1pt}
\paragraph{Proposition~\ref{result:lt_lb}} \hypertarget{pf:lt_lb}{}
\begin{proof}
The speculator can at most buy back using all its ETH. At time $t$, this amount is the solution $\Delta_t$ to the following
$$\frac{\Delta_t \mathcal D}{L_{t-1}+\Delta_t} + N_{t-1} X_t - L_{t-1} - \Delta_t = 0,$$
supposing there is no liquidation at time $t$. It is straightforward to verify the solution, giving the lower bound: 
$$\Delta_t \geq \frac{1}{2} \Big( -\sqrt{\mathcal D^2 - 4\mathcal D \mathcal L_{t-1} + 2\mathcal D N_{t-1} X_t + N_{t-1}^2 X_t^2} + \mathcal D - 2\mathcal L_{t-1} + N_{t-1} X_t\Big).$$
Note that if the speculator is not solvable at time $t$, then there is no real solution.
\end{proof}

\noindent\rule{\textwidth}{1pt}
\paragraph{Proposition~\ref{result:lt_ub}} \hypertarget{pf:lt_ub}{}
\begin{proof}
As above, consider $X_{t+1} = X_t R_{t+1}$. For notational simplicity, we drop subscripts as follows: $\bar N_t \mapsto N$, $X_{t} \mapsto X$, $\mathcal L_t \mapsto \mathcal L$, $\Delta = \mathcal L_t - \mathcal L_{t-1}$, $c(L_t) \mapsto c$, $b(L_t) \mapsto b$, $g_t \mapsto g$, $R_{t+1} \mapsto R$, $\PR(A_t |\mathcal F_t) \mapsto \PR(A)$, $\PR(B_t |\mathcal F_t) \mapsto \PR(B)$.

Suppose the first condition is true. We have
$$\begin{aligned}
\frac{\partial \psi}{\partial \mathcal L}
&= \frac{\mathcal D \mathcal L_{t-1}}{\mathcal L^2} \EX[R | \mathcal F_t]
- \int_{\frac{c}{X}}^\infty g(z) dz
+ \int_{\frac{c}{X}}^{\frac{b}{X}} \left(3-\frac{\alpha \mathcal D N X z}{2(NXz - \mathcal L)^2}\right) g(z) dz\\
&\leq \frac{\mathcal D \mathcal L_{t-1}}{\mathcal L^2} \EX[R | \mathcal F_t] - \PR(A \cup B) \\
&\leq \frac{\mathcal D \mathcal L_{t-1}}{\mathcal L^2} \EX[R | \mathcal F_t] - \kappa^{-1}.
\end{aligned}$$
Notice this is monotonic decreasing in $\mathcal L$ over the domain, so the critical point will be a bound for the optimal value of $\mathcal L^*$. Setting equal to 0, we have
$$\mathcal L^* \leq \sqrt{ \kappa \mathcal D \mathcal L_{t-1} \EX[R|\mathcal F_t]}.$$

Now suppose the second condition is true instead. We have
$$\begin{aligned}
\frac{\partial \psi}{\partial \mathcal L}
&= \frac{\mathcal D \mathcal L_{t-1}}{\mathcal L^2} \EX[R | \mathcal F_t]
- \int_{\frac{b}{X}}^\infty g(z) dz
+ 2\int_{\frac{c}{X}}^{\frac{b}{X}} g(z) dz
- \int_{\frac{c}{X}}^{\frac{b}{X}} \frac{\alpha \mathcal D N X z}{2(NXz - \mathcal L)^2} g(z) dz\\
&\leq \frac{\mathcal D \mathcal L_{t-1}}{\mathcal L^2} \EX[R | \mathcal F_t]
- \Big( \PR(A) - 2\PR(B) \Big) \\
&\leq \frac{\mathcal D \mathcal L_{t-1}}{\mathcal L^2} \EX[R | \mathcal F_t] - \kappa^{-1}.
\end{aligned}$$
which delivers the desired result as above.
\end{proof}

\noindent\rule{\textwidth}{1pt}
\paragraph{Proposition~\ref{result:stable_range}} \hypertarget{pf:stable_range}{}
\begin{proof}
By assuming $T_{Z_0} > \tau$, we have $Z_0 \geq Z_{t\wedge \tau}$. Applying Proposition~\ref{result:lt_ub} to $Z_t = \frac{\mathcal D}{\mathcal L_t}$ provides $Z_{t\wedge \tau} \geq \sqrt{\frac{\mathcal D}{\kappa \mathcal L_{t\wedge \tau -1} r}}$. Notice that the upper bound on $\mathcal L_t$ and the lower bound on $Z_t$ can be written respectively as increasing and decreasing sequences in $t$ starting from initial state as follows:
$$\overline{\mathcal L}_t = (\kappa \mathcal D r)^{\frac{2^t-1}{2^t}} \mathcal L_0^{\frac{1}{2^t}}.$$
$$\underline{Z}_t = \frac{\mathcal{D}}{(\kappa \mathcal D r)^{\frac{2^t-1}{2^t}} \mathcal L_0^{\frac{1}{2^t}}}.$$
These have limits $\overline{\mathcal L}_\infty = \kappa \mathcal D r$ and $\underline{Z}_\infty = \frac{1}{\kappa r}$ that also bound $\mathcal L_t$ and $Z_t$ respectively.
\end{proof}

\noindent\rule{\textwidth}{1pt}
\paragraph{Proposition~\ref{result:stable_submg}} \hypertarget{pf:stable_submg}{}
\begin{proof}
For $t-1 < \tau$,
$$\frac{\mathcal D}{\EX[\mathcal L_t | \mathcal F_{t-1}]}
\leq \EX\left[ \frac{\mathcal D}{\mathcal L_t} | \mathcal F_{t-1}\right]
\leq \frac{\mathcal D}{\mathcal L_{t-1}}$$
by Jensen's inequality and the condition for $\tau>t-1$. Thus we have
$$\EX[\mathcal L_{t\wedge\tau} | \mathcal F_{t-1}] \geq \mathcal L_{t \wedge \tau -1}$$
and $(\mathcal L_{t\wedge \tau})$ is a submartingale. $(Z_{t\wedge\tau})$ is a supermartingale by condition of $\tau$.

Applying Proposition~\ref{result:stable_range}, $\mathcal L_{t\wedge \tau}$ is bounded above and $Z_{t\wedge \tau}$ is bounded below. Thus they converge almost surely by Doob's martingale convergence theorem.
\end{proof}

\noindent\rule{\textwidth}{1pt}
\paragraph{Proposition~\ref{result:stable_exp_bound}} \hypertarget{pf:stable_exp_bound}{}
\begin{proof}
The first inequality follows from Proposition~\ref{result:stable_range} and supermartingale properties.

Since $Z_{t\wedge \tau}$ is supermartingale, we have $Z_{t-1} \geq \EX[Z_t | \mathcal F_{t-1}]$. Assume $(\EX[R_{t+1} | \mathcal F_t])$ is non-decreasing for $t<\tau$. Then subject to the stopping time $\tau$,
$$\begin{aligned}
\EX[Z_t | \mathcal F_{t-1}]
& \geq \EX\left[ \sqrt{
\frac{\mathcal D}{\kappa \mathcal L_{t-1} \EX[R_{t+1} | \mathcal F_t]}
} | \mathcal F_{t-1} \right]
& \hspace{1cm} \text{(Apply Proposition~\ref{result:lt_ub})} \\
&\geq \sqrt{
\frac{\mathcal D}{\kappa \mathcal L_{t-1} \EX\Big[ \EX[R_{t+1} | \mathcal F_t] | \mathcal F_{t-1}\Big]}
}
& \hspace{1cm} \text{(Jensen's inequality)} \\
&= \sqrt{
\frac{\mathcal D}{\kappa \mathcal L_{t-1} \EX[R_{t+1} | \mathcal F_{t-1}]}
}
& \hspace{1cm} \text{(Tower property)} \\
&\geq \sqrt{
\frac{\mathcal D}{\kappa \mathcal L_{t-1} \EX[R_t | \mathcal F_{t-1}]}
}
\end{aligned}$$
since $\EX[R_{t+1}|\mathcal F_t] \geq \EX[R_t | \mathcal F_{t-1}]$.
\end{proof}

\noindent\rule{\textwidth}{1pt}
\paragraph{Lemma~\ref{result:zprime_submg}} \hypertarget{pf:zprime_submg}{}
\begin{proof}
For $t-1 < \tau\wedge T_m$,
$$\begin{aligned}
\EX\left[ | m-Z_t | |\mathcal F_{t-1}\right]
&\geq | \EX[ m-Z_t | \mathcal F_{t-1}]| \\
&\geq | m-Z_{t-1}|,
\end{aligned}$$
by Jensen's inequality and the condition for $t-1 < T_m$ that $m-Z_{t-1} \geq 0$. Thus $\left( Z^\prime_{t\wedge \tau\wedge T_m}\right)$ is a non-negative submartingale.
\end{proof}

\noindent\rule{\textwidth}{1pt}
\paragraph{Proposition~\ref{result:zprime_exp_max}} \hypertarget{pf:zprime_exp_max}{}
\begin{proof}
Note for $t< \tau\wedge T_m$, have  $Z^*_t \leq m$, and so $Z^{\prime *}_{\tau\wedge T_m -1} \leq m - \frac{1}{\kappa r}$. Thus $Z^{\prime *}_{\tau\wedge T_m} \leq \max\Big( m-\frac{1}{\kappa r}, Z^\prime_{\tau\wedge T_m} \Big)$.

Consider time $t=\tau\wedge T_m$ and note that optional stopping applies since $Z$ is bounded. Denote $W:= m-Z_t$, $E:= \EX[-W | Z_t > m]$, and $p:= \PR(Z_t \leq m)$. From optional stopping, we recall that $m \geq \EX[Z_t] \geq \frac{1}{\kappa r}$, and so $0 \leq \EX[W] \leq m-\frac{1}{\kappa r}$. Then
$$\begin{aligned}
\EX[W] &= \EX[ W \Ind_{Z_t \leq m}] - \EX[ - W \Ind_{Z_t > m}] \\
&\leq p\left( m- \frac{1}{\kappa r}\right) - (1-p) E.
\end{aligned}$$
Combining with $0\leq \EX[W]$, we have $0 \leq p(m - \frac{1}{\kappa r}) - (1-p)E$, which gives
$$p \geq \frac{E}{m-\frac{1}{\kappa r} + E}.$$

Then noting that $(1-p)E \leq E(1- \frac{E}{m-\frac{1}{\kappa r} + E})$, $p \leq 1$, and $\EX[Z_t^\prime] = \EX[W \Ind_{Z_t \leq m}] + \EX[-W \Ind_{Z_t>m}]$,
we have
$$\begin{aligned}
\EX[ Z_t^{\prime *}]
&\leq p \EX[ Z_{t-1}^{\prime *}] + (1-p) E \\
&\leq m - \frac{1}{\kappa r} + E\left( 1 - \frac{E}{m-\frac{1}{\kappa r} + E}\right) \\
&= m-\frac{1}{\kappa r} + \frac{E(m - \frac{1}{\kappa r})}{m - \frac{1}{\kappa r} + E}.
\end{aligned}$$
Notice further that given either of the following conditions
\begin{itemize}
\item $\frac{1}{\kappa r} > m$ and $E > \frac{1}{\kappa r} - m$
\item $\frac{1}{\kappa r} = m$ and $E > 0$
\item $\frac{1}{\kappa r} < m$ ad $E \geq 0$,
\end{itemize}
then
$$0 \leq (1-p)E \leq \frac{E(m-\frac{1}{\kappa r})}{m - \frac{1}{\kappa r} + E} \leq m - \frac{1}{\kappa r}.$$
Thus, recalling we used $t=\tau\wedge T_m$, we get the following result
$$\EX[ Z_{\tau \wedge T_m}^{\prime *}] \leq 2\left( m - \frac{1}{\kappa r}\right).$$
\end{proof}

\noindent\rule{\textwidth}{1pt}
\paragraph{Theorem~\ref{result:zprime_doob_ineq}} \hypertarget{pf:zprime_doob_ineq}{}
\begin{proof}
Given Lemma~\ref{result:zprime_submg} and Proposition~\ref{result:zprime_exp_max} and noting $\EX[Z^\prime_{\tau \wedge T_m}] \leq \EX[ Z^{\prime *}_{\tau \wedge T_m}]$, apply Doob's maximal inequality.
\end{proof}

\noindent\rule{\textwidth}{1pt}
\paragraph{Theorem~\ref{result:zprime_qv_pr}} \hypertarget{pf:zprime_qv_pr}{}
\begin{proof}
Apply Theorem~3.1 in \cite{burkholder73}, noting that $\sup_n \EX[ Z^\prime_{n\wedge \tau \wedge T_m}] \leq \EX[ Z^{\prime *}_{\tau \wedge T_m}]$ by Jensen's inequality.
\end{proof}

\noindent\rule{\textwidth}{1pt}
\paragraph{Theorem~\ref{result:unstable_submg}} \hypertarget{pf:unstable_submg}{}
\begin{proof}
For $S_1 \leq t < S_2$, we have
$$\begin{aligned}
\EX\left[\frac{\mathcal D}{\mathcal L_t} | \mathcal F_{t-1}\right]
&\geq \frac{\mathcal D}{\EX[\mathcal L_t |\mathcal F_{t-1}]} \\
&\geq \frac{\mathcal D}{\mathcal L_{t-1}}
\end{aligned}$$
by Jensen's inequality and the $S_1$ condition $\EX[\mathcal L_t |\mathcal F_{t-1} ] \leq \mathcal L_{t-1}$. Thus $(Z_{S_1\vee t \wedge S_2})$ is a submartingale (though note that it can be a submartingale for more general stopping times than this).

$L$ started at $S_1$ and stopped $S_2$ is a supermartingale (by definition).
\end{proof}

\noindent\rule{\textwidth}{1pt}
\paragraph{Theorem~\ref{result:var_approx}} \hypertarget{pf:var_approx}{}
\begin{proof}
As above, consider $X_{t+1} = X_t R_{t+1}$. For notational simplicity, we drop subscripts as follows: $\bar N_t \mapsto N$, $X_{t-1} \mapsto X$ (notice this is different from previous usage), $\mathcal L_t \mapsto \mathcal L$, $\Delta = \mathcal L_t - \mathcal L_{t-1}$, $c(L_t) \mapsto c$, $b(L_t) \mapsto b$, and $g_t \mapsto g$.

Let $\rho$ be (deterministic) variable representing the outcome of $R_t$, such that now we have the outcome $X_t = X\rho$. And define $h(\rho) = \arg\max_L \psi(\rho,L) = \EX[Y_{t+1} | \mathcal F_t]$. By first order condition, $\frac{\partial}{\partial L} \psi(\rho,h(\rho)) = 0$. The assumptions on $\psi$ provide unique maximum and fulfill conditions of the implicit function theorem, which gives us $\frac{\partial h}{\partial \rho}(\rho)$ exists and
$$\frac{\partial h}{\partial \rho}(\rho) = - \frac{\frac{\partial^2}{\partial \rho \partial L} \psi(\rho,h(\rho))}{\frac{\partial^2}{\partial L^2} \psi(\rho,h(\rho))}.$$

Calculating derivatives using the Leibniz integral rule (recalling $c,b$ are functions of $L$),
$$\begin{aligned}
\frac{\partial^2 \psi}{\partial \rho \partial L}
&= g\left( \frac{c}{X\rho}\right) \frac{c}{X\rho^2} \left(4-\frac{\alpha \mathcal D N c}{2(Nc - \mathcal L)^2}\right)
- g\left( \frac{b}{X\rho}\right) \frac{b}{X\rho^2} \left(3-\frac{\alpha \mathcal D N b}{2(Nb - \mathcal L)^2}\right) \\
&\hspace{0.5cm}+ \int_{\frac{c}{X\rho}}^{\frac{b}{X\rho}} \frac{\alpha \mathcal D N X z (NX\rho z + \mathcal L)}{2(NX\rho z-\mathcal L)^3}g(z)dz.
\end{aligned}$$

$$\begin{aligned}
\frac{\partial^2 \psi}{\partial L^2}
&= -\frac{2\mathcal D \mathcal L_{t-1} \EX[R_{t+1}]}{\mathcal L^3}
+ g\left( \frac{b}{X \rho}\right) \frac{\partial b}{\partial L}\frac{1}{X \rho} \left(3-\frac{\alpha \mathcal D N b}{2(Nb - \mathcal L)^2}\right) \\
&\hspace{0.5cm} - g\left( \frac{c}{X\rho}\right) \frac{\partial c}{\partial L}\frac{1}{X\rho} \left(2-\frac{\alpha \mathcal D N c}{2(Nc - \mathcal L)^2}\right)
- \int_{\frac{c}{X\rho}}^{\frac{b}{X\rho}} \frac{\alpha \mathcal D N X \rho z}{(NX\rho z-\mathcal L)^3}g(z)dz.
\end{aligned}$$

Notice that (and continuing with $\beta=3/2$)
$$3-\frac{\alpha \mathcal D N b}{2(Nb - \mathcal L)^2} = 3- \frac{\alpha \mathcal D \beta \mathcal L}{2(\mathcal L (\beta-1))^2} = 3 - \frac{3\alpha \mathcal D}{\mathcal L} < 0,$$
by assumption that liquidation repurchase price always $\geq 1$. And
$$\begin{aligned}
\frac{\alpha \mathcal D N c}{2(Nc-\mathcal L)^2}
&\leq \frac{\frac{1}{2} \alpha \mathcal D(2\alpha \mathcal D + L - \alpha \mathcal D + L)}{-2\alpha \mathcal D(2\alpha \mathcal D +L) + 2L(\alpha \mathcal D + L) + 2\alpha^2\mathcal D^2 + 2\alpha \mathcal D L + 2L^2} \\
&= \frac{\alpha \mathcal D(\alpha \mathcal D + 2L)}{4(\alpha \mathcal D + L)(2L-\alpha \mathcal D)} \\
&= \frac{\alpha\mathcal D}{12(\alpha\mathcal D + L)} + \frac{\alpha \mathcal D}{3(2L-\alpha\mathcal D)} \\
&\leq \frac{1}{12} + \frac{\alpha\mathcal D}{3(2L-\alpha\mathcal D)}.
\end{aligned}$$
This is $\leq2$ when $L \geq \frac{27}{46}\alpha\mathcal D$. Thus under this condition
$$4-\frac{\alpha \mathcal D N c}{2(Nc-\mathcal L)^2} > 2- \frac{\alpha \mathcal D N c}{2(Nc-\mathcal L)^2} \geq 0.$$

Note that all terms of $\frac{\partial^2 \psi}{\partial \rho \partial L}$ are non-negative and all terms of  $\frac{\partial^2 \psi}{\partial L^2}$ are non-positive. Given $\rho \geq b/X$, we have $g\left(\frac{c}{X\rho}\right)$ and $g\left(\frac{b}{X\rho}\right)$ are increasing in $1/\rho$. Note also that $\frac{\partial b}{\partial L}$, $\frac{\partial c}{\partial L}$, and $\frac{2\mathcal D \mathcal L_{t-1} \EX[R_{t+1}]}{\mathcal L^3}$ are constant in $\rho$. Lastly, the numerator and denominator integrals can be rewritten respectively as
$$\frac{1}{\rho} \int_{c}^{b} \frac{\alpha \mathcal D N z(Nz + \mathcal L)}{2(Nz-\mathcal L)^3}g\left(\frac{z}{X\rho}\right)dz
\text{ and } \int_{c}^{b} \frac{\alpha \mathcal D N z}{(Nz-\mathcal L)^3}g\left(\frac{z}{X\rho}\right)dz$$
and $\frac{\alpha \mathcal D N z(Nz + \mathcal L)}{2(Nz-\mathcal L)^3} \geq \frac{\alpha \mathcal D N z}{(Nz-\mathcal L)^3}$ given $Nz+\mathcal L \geq Nc + \mathcal L > 2$, for which $\mathcal L > 8$ is sufficient. And so the terms in the numerator of $|h^\prime(\rho)|$ are growing by a factor $1/\rho$ faster than the terms in the denominator as $\rho$ decreases, proving (2).

Next, note that under the condition $0<\rho<1$,
$$\frac{b}{X\rho^2} = \frac{\beta L}{NX\rho^2} = \frac{db}{dL} \frac{L}{X\rho^2} \geq \frac{db}{dL} \frac{1}{X\rho}$$
$$\frac{c}{X\rho^2} \geq \frac{dc}{dL} \frac{1}{X\rho^2} \geq \frac{dc}{dL} \frac{1}{X\rho}.$$
The last relation uses the fact that $\frac{dc}{dL} \leq \frac{2\alpha \mathcal D + L}{2(\alpha\mathcal D + L)} + 1 < 2$, and so $c>\frac{dc}{dL}$ under the problem setup.

Next note that for $\rho \leq \frac{\mathcal L}{8}$ and $c \leq X\rho z \leq b$, we have
$$\frac{\alpha \mathcal D N Xz(NX\rho z+\mathcal L)}{2(NX\rho z-\mathcal L)^3} \geq \frac{\alpha DN X\rho z}{(NX\rho z-\mathcal L)^3}.$$
This is because the expression (1) simplifies to $NX\rho z + \mathcal L \geq 2\rho$, (2) to be true over the whole range of $z$, we need $Nc + \mathcal L \geq 2\rho$, and (3) $\rho\leq \frac{\mathcal L}{8}$ is sufficient for this. Thus
$$\int_{\frac{c}{X\rho}}^{\frac{b}{X\rho}} \frac{\alpha \mathcal D N X z (NX\rho z + \mathcal L)}{2(NX\rho z-\mathcal L)^3}g(z)dz
\geq \int_{\frac{c}{X\rho}}^{\frac{b}{X\rho}} \frac{\alpha \mathcal D N X \rho z}{(NX\rho z-\mathcal L)^3}g(z)dz$$
under these conditions.

Then note that all terms in the numerator of $h^\prime(\rho)$ are greater than and grow faster in $1/\rho$ than the comparable terms in the denominator. This leaves the first term in the numerator, which is constant in $\rho$. To get (3), then note that $\varepsilon$ can be chosen such that for $\rho =\varepsilon$, the numerator and denominator are equal.

We can derive the results for $\frac{\partial h}{\partial n}$ in essentially the same way. Alter the above dropping of subscripts with $X_t \mapsto X$, let $n$ be a variable representing the realization of $\bar N_t$, and consider $h$ as a function of $n$. Note the following relevant derivatives.

$$\frac{\partial b}{\partial n} = -\frac{\beta L}{n^2} = - \frac{b}{n}$$

$$\frac{\partial c}{\partial n} = -\frac{1}{2n^2} \Big( \sqrt{\alpha^2 \mathcal D^2 + 4\alpha\mathcal D L + L^2} -\alpha\mathcal D + L\Big) = - \frac{c}{n}$$

$$\begin{aligned}
\frac{\partial^2 \psi}{\partial n \partial L}
&= g\left( \frac{c}{X}\right) \frac{c}{n} \left(2-\frac{\alpha \mathcal D n c}{2(nc - \mathcal L)^2}\right)
- g\left( \frac{b}{X}\right) \frac{b}{n} \left(3-\frac{\alpha \mathcal D n b}{2(nb - \mathcal L)^2}\right) \\
&\hspace{0.5cm}+ \int_{\frac{c}{X}}^{\frac{b}{X}} \frac{\alpha \mathcal D n X z (nX z + \mathcal L)}{2(nX z-\mathcal L)^3}g(z)dz.
\end{aligned}$$
And translating the following to the new notation
$$\begin{aligned}
\frac{\partial^2 \psi}{\partial L^2}
&= -\frac{2\mathcal D \mathcal L_{t-1} \EX[R_{t+1}]}{\mathcal L^3}
+ g\left( \frac{b}{X }\right) \frac{\partial b}{\partial L}\frac{1}{X} \left(3-\frac{\alpha \mathcal D n b}{2(nb - \mathcal L)^2}\right) \\
&\hspace{0.5cm} - g\left( \frac{c}{X}\right) \frac{\partial c}{\partial L}\frac{1}{X} \left(2-\frac{\alpha \mathcal D n c}{2(nc - \mathcal L)^2}\right)
- \int_{\frac{c}{X}}^{\frac{b}{X}} \frac{\alpha \mathcal D n X z}{(nX z-\mathcal L)^3}g(z)dz.
\end{aligned}$$
And by applying implicit function theorem, we get
$$\frac{\partial h}{\partial n}(n) = - \frac{\frac{\partial^2}{\partial n \partial L} \psi(n,h(n))}{\frac{\partial^2}{\partial L^2} \psi(n,h(n))}.$$
From here we can proceed with the same analysis using factors of $\frac{1}{n}$ instead of $\frac{1}{\rho}$.
\end{proof}

\noindent\rule{\textwidth}{1pt}
\paragraph{Theorem~\ref{result:unstable_var}} \hypertarget{pf:unstable_var}{}
\begin{proof}
For notational simplicity, drop subscripts $X_t \mapsto X$, $\bar N_{t-1} \mapsto N$, $\mathcal L_{t-1} \mapsto \mathcal L$. And consider $x$ a realization of $X$ as variable in $h$. Define the function $f(X,n) = \frac{1}{h(X,n)}$ where $n$ represents the realization of $N$. With probability 1, the following are true:

\begin{itemize}
\item $h$ is concave in $x$ and $n$ because $h^\prime$ is decreasing, as shown in the previous result.

\item $f$ is differentiable (wrt $n$ and $x$) over domain using chain rule and implicit function theorem.

\item $f$ is convex:  it's the composition of $1/x$ and $h$, and since $1/x$ is convex and non-increasing and $h$ is concave, so is $f$ (see \cite{boyd09} 3.2.4).

\item $f$ is (strictly) decreasing (in $n$ and $x$) since $h$ is increasing.

\item By assumption, we've restricted $NX$. The derivative of $f$ at the minimum value exists and is bounded.

\item $f$ is non-negative since $h$ is non-negative.

\item $\frac{\partial f}{\partial n}$ is (strictly) increasing in $n$. We have
$$f^\prime (x,n) = -\frac{1}{h(x,n)^2} h^\prime(x,n),$$
where $h^\prime(x,n)$ is derived in the previous proof using the implicit function theorem. $h$ is increasing in $n$ and $h^\prime$ is non-negative and decreasing in $n$. Thus $\frac{h^\prime}{h^2}$ is decreasing in $n$, and so $-\frac{h^\prime}{h^2}$ is increasing.

\item $\frac{\partial h}{\partial n}$ is increasing in $x$. This can be seen using the formulation at the end of the proof for the previous result as terms in $\frac{\partial^2\psi}{\partial L^2}$ grow slower in $x$ (in magnitude) than terms in $\frac{\partial^2\psi}{\partial n \partial L}$. In particular, the first term of $\frac{\partial^2\psi}{\partial L^2}$ is decreasing in magnitude since $L$ is increasing in $x$. And the integral in $\frac{\partial^2 \psi}{\partial n \partial L}$ increases faster in $x$ than the integral in $\frac{\partial^2\psi}{\partial L^2}$, as can be seen by comparing the integrand numerators (a factor of $x^2$ in $\frac{\partial^2 \psi}{\partial n \partial L}$ vs. a factor of $x$ in $\frac{\partial^2\psi}{\partial L^2}$).

\item $\frac{\partial f}{\partial n}$ is (strictly) increasing in $x$ This is because $h$ is increasing in $x$ and $\frac{\partial h}{\partial n}$ is non-negative and increasing in $x$ (previous bullet).
\end{itemize}

Note additionally that, from the system setup assumptions, all of the functions are appropriately bounded.

Thus we can apply Theorem~3.1 in \cite{see08} to get
$$\text{Var}\Big( f(X,N^s) | \mathcal F_{t-1} \Big) < \text{Var}\Big( f(X,N^u) | \mathcal F_{t-1} \Big).$$
Note that the variances exist because $h = \mathcal L_t$ is bounded, as shown in previous results. The variances of $Z_t^s$ and $Z_t^u$ are then obtained by multiplying the above inequality by $\mathcal D^2$.
\end{proof}